\documentclass[11pt,a4paper]{article}

\newcommand{\web}[2]{#1}

\usepackage{amsmath,amssymb}
\usepackage[dvipsnames]{xcolor}
\usepackage{graphicx}
\usepackage{enumerate}
\usepackage{url}
\usepackage{hyperref}
\usepackage{ntabbing}
\usepackage{soul}

\web{\usepackage {ifthen}

\ifthenelse{\isundefined{\definition}}{\newtheorem{definition}{Definition}}{}

\ifthenelse{\isundefined{\theorem}}{\newtheorem{theorem}{Theorem}}{}
\ifthenelse{\isundefined{\conjecture}}{\newtheorem{conjecture}{Conjecture}}{}
\ifthenelse{\isundefined{\lemma}}{\newtheorem{lemma}{Lemma}}{}
\ifthenelse{\isundefined{\corollary}}{\newtheorem{corollary}{Corollary}}{}
\ifthenelse{\isundefined{\example}}{\newtheorem{example}{Example}}{}
\ifthenelse{\isundefined{\property}}{\newtheorem{property}{Property}}{}
\ifthenelse{\isundefined{\proof}}{\newenvironment{proof}{\noindent{\bf Proof.}\hspace{0.05in}}{\par\mbox{}\par}}{}

\newcommand{\vs}{\vspace*{0.3cm}}

\ifthenelse{\isundefined{\qed}}{\newcommand{\qed}{\mbox{$\blacksquare$}}}{}

\hyphenation{}

\setlength{\unitlength}{6pt}

\newenvironment{lists}[1]{
                 \begin{list}{}{
                     \setlength{\listparindent}{0in}
                     \settowidth{\labelwidth}{#1}
                     \setlength{\leftmargin}{\labelwidth}
                     \addtolength{\leftmargin}{\labelsep}
                     }
                 }{
                 \end{list}
                 }

\newenvironment{given-find}[2]{
                               \vs 
                               \noindent \hrule
                               \begin{lists}{Given:XX}
                               \item[\sc Given: \hfill] #1                                 
                               \item[\sc Find: \hfill] #2                               
                               \vs 
                               \noindent \hrule 
                               }{
                               \end{lists}
                               }

}{}
\newcommand{\citeyear}[1]{\cite{#1}}

\newcommand{\True}{\mathbf{T}}
\newcommand{\False}{\mathbf{F}}
\newcommand{\I}{\mathcal{I}}
\newcommand{\IEvalT}[2]{\I(#1)(#2)} % interval evaluation of terms

\newcommand{\R}{\mathbb{R}}
\newcommand{\C}{\mathbb{C}}
\newcommand{\N}{\mathbb{N}}

\renewcommand{\deg}{\mathrm{deg\,}}

\newcommand{\Width}[1]{\mathrm{width}(#1)}
\newcommand{\Subst}[3]{#1[#2\leftarrow #3]}

\newcommand{\acktext}{The work of Stefan Ratschan and Peter Franek was
supported by M{\v S}MT project number OC10048 and the Czech Science Foundation (GACR) grants number P202/12/J060 and 15-14484S with institutional support RVO:67985807.}
\newcommand{\ExtVer}[1]{\web{\footnote{#1}}{\thanks{#1}}}

\renewcommand{\epsilon}{\varepsilon}

\newcommand{\Long}[1]{} %{{\color{blue}#1}}              % mainly used for documenting internal
\title{Quasi-decidability of a Fragment of the\\ First-order Theory of Real
  Numbers\ExtVer{This is an extended and revised version of a paper that appeared in the
  proceedings of the 36th International Symposium on Mathematical Foundations of
  Computer Science~\cite{Franek:11}. \web{\acktext}{}}}
\web{\author{PETER FRANEK and STEFAN RATSCHAN\footnote{ORCID: 0000-0003-1710-1513}\\Institute of Computer Science, Academy of Sciences of the Czech Republic\\and\\
PIOTR ZGLICZYNSKI\\Jagellonian University in Krakow}}{\author{Peter Franek \and Stefan Ratschan \and Piotr Zgliczynski}
\institute{Peter Franek \and Stefan Ratschan (ORCID: 0000-0003-1710-1513)\at Institute of Computer Science, Academy of Sciences of the Czech Republic \and Piotr Zgliczynski \at Jagellonian University in Krakow}}

\begin{document}
%\category{}{}{}
%\terms{}
%\keywords{}

\maketitle

\begin{abstract}
In this paper we consider a fragment of the first-order theory of the real numbers that includes systems of $n$ equations in $n$ variables, and for which all functions are computable in the sense that it is possible to compute arbitrarily close interval approximations. Even though this fragment is undecidable, we prove that---under the additional assumption of bounded domains---there is a (possibly non-terminating) algorithm for checking satisfiability such that (1) whenever it terminates, it computes a correct answer, and (2) it always terminates when the input is robust. A formula is robust, if its satisfiability does not change under small continuous perturbations. We also prove that it is not possible to generalize this result to the full first-order language---removing the restriction on the number of equations versus number of variables. As a basic tool for our algorithm we use the notion of degree from the field of topology.
\end{abstract}

\section{Introduction}

It is well known that, while the theory of real numbers with addition and multiplication is decidable~\cite{Tarski:51}, any periodic function makes the problem undecidable, since it allows encoding of the integers. The root existence problem for uni-variate functions defined by addition, multiplication, the sine function and the constant $\pi$ is also undecidable~\cite{Wang:74}. This even holds if we consider only functions on bounded domains, because an algorithm deciding it could be used to compute a fixed point of a continuous function from a ball to itself which is known to be non-computable for some computable functions~\cite{Baigger:85,Potgieter:08}.

Recently, several papers~\cite{Fraenzle:99,Ratschan:02b,Ratschan:02f,Damm:07} have argued, that in continuous domains  (where we have notions of neighborhood, perturbation etc.) such undecidability results do not always have much practical relevance. The reason is, that real-world manifestations of abstract mathematical objects in such domains will always be exposed to perturbations (imprecision of production, engineering approximations, unpredictable influences of the environment etc.). Engineers take these perturbations into account by coming up with \emph{robust} designs, that is, designs that do not change essentially under such perturbations. Hence, in this context, it is sufficient to come up with algorithms that are able to decide such robust problem instances. They are allowed to run forever in non-robust cases, but \emph{must not} return incorrect results, in whatever case. In a recent paper we called problems possessing such an algorithm \emph{quasi-decidable}~\cite{Ratschan:10a}. 

% The main contribution of this paper can be summarized as follows:
% \begin{itemize}
% \item We analyze quasi-decidability of systems of equations containing {arbitrary} interval-computable functions (see Def.~\ref{def:intcomp})
% \item We consider more general formulas containing inequalities, quantifiers, conjunctions and disjunctions. On the positive side, we show
% quasi-decidability of a fragment of the first-order theory of real numbers (Theorem 1). On the negative side, we show that this fragment can not be extended to the full first-order theory (Theorem 2).
% \end{itemize}
% }

The main contribution of this paper can be summarized as follows:
\begin{itemize}
\item We show quasi-decidability of a certain fragment of the first-order theory of the reals (Theorem~\ref{thm:main}). The basic building blocks are existentially quantified disjunctions of systems of $n$ equalities over at most $n$ variables and arbitrarily many inequalities. Those blocks may be combined using universal quantifiers, conjunctions, and disjunctions. All variables are assumed to range over closed and bounded intervals.
\item We show that the result cannot be extended to the full first-order language. More specifically, in the basic building blocks (systems of equalities and inequalities) it is impossible to remove the restriction that the number of variables has to be at most the number of equalities (Theorem~\ref{thm:nonquasidec}). Still, while we show that this restriction cannot be removed completely, this leaves open the possibility to replace the restriction by a weaker constraint on the number of variables and equations.
\end{itemize}

The allowed function symbols include addition, multiplication, exponentiation, and sine. More specifically, they have to be continuous, and for compact intervals $I_1,\dots,I_n$, we need to be able to compute an interval $J\supseteq f(I_1\times \dots \times I_n)$ such that the over-approximation of $J$ over $f(I_1\times \ldots \times I_n)$ can be made arbitrarily small.

The main tool we use is the notion of the {\it degree of a continuous function} that comes from differential topology. 
For continuous functions $f: [a, b]\rightarrow \R$, the degree $\deg(f, [a, b], 0)$ is $0$ iff $f(a)$ and $f(b)$ have the same sign, otherwise the degree is either $1$ or $-1$, depending on whether the sign changes from negative to positive or the other way round. 
If $f$ is continuous and the degree is nonzero, then the equation $f(x)=0$ has a solution  by the intermediate value theorem.
For higher dimensional functions, the degree is a computable~\cite{Aberth:94,Franek:12b} integer whose value may be greater than $1$, 
and a nonzero degree still indicates the existence of a root of $f$.  
The converse is not true and the existence of a root does not imply nonzero degree in general. We show how, for robustly satisfiable formulas built up from certain blocks of $n$ equations in $n$ variables, to make the degree test eventually succeed, while at the same time handling inequalities and logical symbols. 

The proof of our second contribution---the class of equations and inequalities with no relation between the number of equations and variables is \emph{not} quasi-decidable---is based on a reduction from a~recent undecidability result~\cite{Franek:2014} for a~related robust satisfiability problem, cited in Theorem~\ref{thm:non-dec}.

Even though this work applies results from a quite distant field---topology---to
automated reasoning, the paper is largely self-contained. Usage of results from
topology that are not explicitly delineated in this paper is concentrated exclusively in
Section~\ref{technical_section}.

The content of the paper is as follows: In Section~\ref{sec:main-theorem}, we define the notions of robustness and quasi-decidability, and state the two main theorems of the paper.  In Section~\ref{problem}, we provide the quasi-decision procedure whose existence is claimed by the first main theorem. In Section~\ref{degdef}, we present the notion of topological degree and describe its main properties. In Section~\ref{sec:correctness}, we show that the quasi-decision procedure always returns a correct result. In Section~\ref{technical_section} we show some
non-algorithmic properties of the degree that will be the essential for showing termination for robust inputs in Section~\ref{sec:definiteness}. In Section~\ref{sec:nonquasidec} we prove the second main theorem. In Section~\ref{sec:related-work} we discuss related work. Finally, in Section~\ref{sec:conclusion}, we conclude the paper.

\section{Statement of the Results}
\label{sec:main-theorem}
We will start this section with informal discussion of a motivating example. Consider the first-order predicate logic formula \[ \exists x \;.\; [x\geq -1 \wedge x\leq 1 \wedge \sin x= 0 ] \] with the usual interpretation over the real numbers. This formula is true, and remains true, even if it is perturbed a little bit. On the other hand, the formula \[ \exists x \;.\; [ x\geq 1 \wedge x\leq 2\wedge \sin x= 1] \] is also true, but does not remain true when perturbing it, for example by increasing the right-most number $1$ a little bit. We will later call formulas of the first type robust, and formulas of the second type non-robust. Our first theorem will state that, for a certain class of formulas over the reals that includes function symbols such as $\sin$, there exists an algorithm (a ''quasi-decision procedure'') that decides whether a given formula is true, but that is only required to terminate for robust inputs while it may run forever for non-robust inputs. 

In the rest of the section, after fixing notation, we define the class of functions that we consider (Definition~\ref{def:intcomp}). Then we will formalize the notion of perturbing predicate-logical formulas (Definition~\ref{dist}) which results in a precisely defined notion of a formula being robust (Definition~\ref{robust}). Finally, we state Theorem~\ref{thm:main} that ensures the existence of such a quasi-decision procedure and the negative Theorem~\ref{thm:nonquasidec} that puts a limit on generalization of the approach.

We define a {\it box} in $\R^n$ (or also $n$-box) to be the Cartesian product of $n$ closed intervals of finite length (i.e., a hyper-rectangle). The \emph{width} $\Width{B}$ of a box $B$ is the maximum of the width of the constituting intervals of $B$.
For $x\in\R^n$, $|x|$ will refer to its \emph{maximum norm} $|x|:=\max \{|x_1|,\ldots, |x_n|\}$ and for a continuous function
$f:\Omega\to\R^n$, we use the supremum norm $||f||_\Omega:=\sup
\{|f(x)|;\,\,x\in\Omega\}$. 
If $||f-g||_\Omega\leq\alpha$ for some $\alpha>0$, we say that $g$ is an \emph{$\alpha$-perturbation} of $f$ in $\Omega$. If $\Omega$ is clear from the context then we will simply write $||f||$, or say that $g$ is an \emph{$\alpha$-perturbation} of $f$, without explicitly mentioning $\Omega$.
For a set $\Omega\subseteq\R^n$, $\bar\Omega$ is its closure, $\Omega^\circ$ its interior and 
$\partial\Omega=\bar\Omega\backslash\Omega^\circ$ 
its boundary with respect to the Euclidean topology. 
We will call the closure $\bar\Omega$ of an open connected bounded set $\Omega$ a \emph{closed region}. 
%If $\Omega$ is a $k$-box in $\R^n$, we usually
%denote $\partial\Omega$ the boundary in the topology of $\Omega$ (i.e., union of
%the $2k$ faces).

For defining the class of formulas, we will first fix the class of functions that we handle. Intuitively, we allow functions whose range can be arbitrarly closely approximated by boxes:

%\marginpar{\Peter{I took your version, except for the things in red}}
\begin{definition}
\label{def:intcomp}
Let $\Omega\subseteq\R^m$ be a box with rational vertices. 
We say that a function $f: \Omega\to\R^n$ is \emph{interval computable}, iff there exists a corresponding algorithm $\I(f)$ that computes, for any box $B\subseteq\Omega$ with rational vertices, an $n$-box $\I(f)(B)\subseteq\R^n$ with rational vertices such that
\begin{itemize}
\item $\I(f)(B)\supseteq \{ f(x) \mid x\in B\}$, and
\item for every $\varepsilon>0$ there is a $\delta>0$ such that for every box $B$
  with $0<\Width{B}<\delta$, $\Width{\IEvalT{f}{B}}<\varepsilon$.
\end{itemize}
%\marginpar{\Peter{I am not sure whether it is possible to remove this last sentence of the definition.}}
%\Peter{We call a function $f: \R^n\rightarrow\R$ interval-computable iff every restriction of $f$ to a box $\Omega\subseteq\R^n$ with rational endpoints is interval-computable.}
\end{definition}

Each interval computable function is uniformly continuous. Moreover, a function $f: \Omega\to\R^n$, with $\Omega\subseteq\R^m$ a box with rational vertices, is interval computable iff it is computable in the sense of computable analysis~\cite{Brattka:08} (for seeing this, note especially that a function that is computable in the sense of computable analysis has a computable modulus of continuity~\cite[Theorem 2.13]{Ko:91}). 

For common function symbols that can be written in terms of symbolic expressions containing symbols denoting rational constants, the constant $\pi$, addition, multiplication, exponentiation, trigonometric functions and square root, the algorithm  $\I(f)$ can be implemented from the expression by interval arithmetic~\cite{Neumaier:90,Moore:09} with arbitrary precision interval endpoints.

In the rest of the paper, we assume that a set of function and predicate symbols is given, together with structure assigning to each function symbol an interval computable function and to each predicate symbol a corresponding relation over the real numbers. We assume that this symbol set contains at least all rational constants, addition, multiplication, and the predicate symbols $=$ and $\geq$ with their usual interpretation. Whenever we will write concrete function or predicate symbols, this structure will assign their standard meaning over the real numbers. From now on, we will restrict ourselves to formulas from the first-order language corresponding to the given symbol set.

We also assume that a map $\I$ is given that assigns, to each function symbol $f$, an algorithm $\I(f)$ satisfying the specification in Definition~\ref{def:intcomp}. This map $\I$ is assumed to be algorithmic. Such assignment $\I$ naturally extends to terms of the language
via composition of interval functions: if $t$ is a term of the language, then the algorithm $\I(t)$ 
represents the corresponding function and satisfies both assumptions of Definition~\ref{def:intcomp}. In addition, we will assume that every variable ranges over a closed bounded interval introduced by a corresponding quantifier of the form $\exists x\in I$ or $\forall x\in I$. Throughout the paper we will require those bounds to be small enough to avoid any function application outside of the domain of any interval computable function. In a similar way, whenever we introduce bounds on the free variables of a formula, we assume them to be small enough to avoid such function applications. %we use the formulation "a box bounding the free variables of the formula"

As usual, a sentence will refer to a formula without free variables.  Now we formalize perturbations of formulas by defining some notion of distance on sentences. 

\begin{definition}
\label{dist}
  Let $F, G$ be two sentences. We say that $F$ and $G$ have the \emph{same structure} iff one can be obtained from the other
  by only exchanging terms (i.e., they have the same Boolean and quantification
  structure including bounds of quantified variables, and the same predicate symbols).

We define the distance $d$ on sentences as follows. If two sentences $F$ and $G$
do not have the same structure, then
$d(F, G):=\infty$. In the case where they do have the same structure, assume that the sentence $F$ contains
terms denoting functions $f_1,\ldots,f_p$ and the sentence $G$ contains in the corresponding 
places terms denoting the functions $g_1,\ldots,g_p$. We define the distance 
$$d(F, G):=\max_{i\in\{1,\ldots,p\}} ||f_i-g_i||_{\Omega_i},$$ where $\Omega_i$
denotes the respective domain of those functions, that is, the box defined by the quantification of all the variables.
\end{definition}

For example, the sentences 
$$\exists x\in [0,1]\,\forall y\in [0,1] \;.\; x^2-y=xy \wedge x=y$$ and 
$$\exists x\in [0,1]\,\forall y\in [0,1] \;.\; x^2-y=xy+1 \wedge x=y^2$$ 
have the same structure, because the only difference is in the terms involved.
The distance $d(F,G)=1$, because---with $(x,y)\in [0,1]^2$---we have that $\max |(x^2-y) -  (x^2-y)|=0$, $\max |xy - (xy+1)|=1$, $\max |x-x|=0$ and $\max |y-y^2|=1/4$. As another example, the sentences $1\geq 0$ and $\neg\neg 1\geq 0$ do not have the same structure, and hence their distance is $\infty$.

% \begin{definition}
% \label{robust}
% Let S be a sentence in class $\mathcal{A}$ and $\varepsilon>0$.  We say that S is
% \emph{$\varepsilon$-robust}, if for every sentence $S'$, $d(S',S)<\varepsilon$ implies
% that $S'$ and $S$ are equi-satisfiable. 
% We say that the sentence $S$ is \emph{robust}, if there is an $\varepsilon>0$ such
% that $S$  is $\varepsilon$-robust.
% %Any $\varepsilon>0$ such that $S$ is $\varepsilon$-robust will be called a "bound of robustness" of $S$.
% We say that a sentence  $S$ is robustly true, if it is both robust and true
% (i.e., satisfiable).
% We say that a sentence $S$ is robustly false, if it is both robust and false
% (i.e., unsatisfiable).
% \end{definition}

%If $f,g: K\to\R^n$ are continuous, then $d(\bigwedge_i f_i=0, \bigwedge_i g_i=0)=\max_i ||f_i-g_i||=||f-g||$, which justifies the use of the max-norm in $\R^n$. \marginpar{\Peter{I do not fully get the point of this comment.}}

\begin{definition}
\label{robust}
Let $S$ be a sentence and $\varepsilon>0$.  We say that $S$ is \emph{$\varepsilon$-robust} iff for every sentence $S'$, $d(S',S)<\varepsilon$ implies that $S'$ and $S$ have the same truth value. We say that the sentence $S$ is \emph{robust} iff there is an $\varepsilon>0$ such that $S$  is $\varepsilon$-robust.
%Any $\varepsilon>0$ such that $S$ is $\varepsilon$-robust will be called a "bound of robustness" of $S$.
We say that a sentence  $S$ is \emph{robustly true} iff it is both robust and true.
We say that a sentence $S$ is \emph{robustly false} iff it is both robust and false.
\end{definition}
Note that, since we restricted ourselves to formulas with function symbols denoting interval-computable functions, all functions involved in the above definitions are interval computable, hence uniformly continuous.

Also note that equivalence of two formulas does not necessarily imply the same robustness. For example, the formula $\exists x \in [ 0, 2] \;.\; x-1=0$ is robust, but the formula $\exists x \in [ 0, 2] \;.\; x-1=0 \wedge x-1=0$ is not, since both occurrences of the function $x-1$ can be perturbed independently. 
%This is natural, since the equivalence of two terms containing function symbols like $\sin$ cannot be algorithmically checked~\cite{Richardson:68,Caviness:70}.

% \begin{definition}
% \label{def:quasidec}
%   Given a class $\mathcal{X$ of first-order sentences and a function $d:
%   \mathcal{B}\times\mathcal{B}\rightarrow \mathbb{R}^{\geq 0}\cup \{\infty\}$,
%  $\mathcal{B}$ iss
%   \emph{quasi-decidable wrt. $d$} iff

% \end{definition}

\begin{definition}
\label{def:quasidec}
A \emph{quasi-decision procedure} for some class $\mathcal{B}$ of formulas is an algorithm that takes as inputs a \emph{sentence} $\varphi$ from $\mathcal{B}$ 
and an algorithm $\I$ converting function symbols $f$ to algorithms $\I(f)$.
The algorithm computes the truth value of $\varphi$ whenever $\varphi$ is robust. 
If $\varphi$ is non-robust, the algorithm may run forever but must not return an incorrect result.

If such a quasi-decision procedure exists for some class $\mathcal{B}$, then we say that $\mathcal{B}$ is \emph{quasi-decidable}.
\end{definition}

Now we are ready to state our first result.
\begin{theorem}
\label{thm:main}
The following class of formulas $\mathcal{B}$, defined recursively below, is quasi-decidable:
\begin{enumerate}[(a)]
\item $\mathcal{B}$ contains all formulas of the form
\label{item:Bbasic}
\[
\begin{array}{lc}
\exists x\in B \;.\; & [f_1=0\, \wedge\, f_2=0\,\wedge\, \ldots\wedge\,
f_{n}=0\,\wedge\, g_1\geq 0\,\wedge\, g_2\geq 0\,\wedge\ldots\wedge\,
g_{k}\geq 0 ]
\end{array}\]
where $f_1, \dots, f_n, g_1,\dots, g_k$ are terms denoting interval-computable functions, $B$ is an $m$-box (the expression $\exists x\in
B$ denoting a block of $m$ existential quantifiers) with rational vertices and either $n\geq m$ or $n=0$.
The integer $k$ may be arbitrary and we also admit $k=0$ (i.e., the case without
inequalities). 
%  which representes the case where $g^\alpha_j$ are omitted and only equalities are present. 
% The special case $n_\alpha=0$ corresponds to a system of $k$ inequalities without $f^\alpha_i$ in the particular parenthesis.

\item \label{item:Buniv}
Let $I\subseteq \R$ be a closed bounded interval with rational endpoints. If $U$ is in $\mathcal{B}$, then
$$\forall x\in I \;.\; U$$ is also in $\mathcal {B}$.

\item \label{item:Bconn} If $U,V$ are in $\mathcal {B}$, then $$U\wedge V\quad\ U\vee V$$ are also in $\mathcal{B}$.

%(d) Let $I$ be a bounded closed interval, $U$ and $V$ are two formulas in $\mathcal{B}$ 
%containing the variable $x$ only freely. If $$\exists x\in I\,\, U\quad{\rm{and}}\quad \exists x\in I\,\, V$$ 
%are in $\mathcal{B}$, then
%$$\exists x\in I\,\, (U\vee W)$$ is also in $\mathcal{B}$.

\end{enumerate}
\end{theorem}
The formulas corresponding to $(\ref{item:Bbasic})$ represent systems of
equations and inequalities.  However, we assume that
there are no more existential quantifiers than equations in~$(\ref{item:Bbasic})$, corresponding to the condition $n\geq m$.

The following sentence is an example of a formula in class $\mathcal{B}$:

\vspace*{0.4cm}
$\forall x\in [-1,1]$\\ 
\hspace*{1cm} $\exists y\in [-1,1]\,\, \exists z\in [-1,1]$\\
\hspace*{2cm} $[x^2-y^2-z^2=0\,\, \wedge\,\, x^3-y^3-z^3=0]$.
\vspace*{0.4cm}

% because the large parenthesis behind $\forall x$ is a formula of type $(a)$ (two equations and one inequality)
% and adding th universal quantifier corresponds to applying the case $(b)$. 

The following sentence is an example of a sentence not in $\mathcal{B}$
$$
\exists x\in [0,1]\,\,\exists y\in [0,1] \;.\; x-y=0
$$
because the domain of the particular function is a $2$-dimensional box and there is only one equation, so the assumptions in $(\ref{item:Bbasic})$ are violated.

Throughout we will use the convention that logical connectives bind stronger
than quantifiers. Moreover, we use brackets to denote Boolean structure of
formulas. Sometimes we will use line breaks instead of brackets for this
purpose. We will use the symbol $\equiv$ to denote equality of first-order
formulas. 

If $\exists x\in B\;.\;\,F_1$ and $\exists x\in B\;.\;\,F_2$ are in the class
$\mathcal{B}$, then $\exists x\in B \;.\; [F_1\,\vee F_2]$ is robust if and only if 
the formula $[\exists x\in B\;.\; F_1]\,\vee\,[\exists x\in B\;.\; F_2]$ is robust
and they are equi-satisfiable. Hence a quasi-decision procedure for $\mathcal{B}$ can
handle disjunctions within existential quantification, too.
%Also, note that a formula with strict inequalities of the form $\exists x\in B \;.\; f=0\,\wedge g>0$ is
%robust if and only if $\exists x\in B \;.\; f=0\,\wedge g\geq 0$ is robust and they are equi-satisfiable. 
In the following, however, we will restrict ourselves to the class $\mathcal{B}$.

The following theorem shows a limitation of possible extension of quasi-decidability of the class $\mathcal{B}$ to the whole first-order theory removing the restriction on the number of equations versus number of variables:
\begin{theorem}
\label{thm:nonquasidec}
Assume that the our symbol set is rich enough to contain function symbols for all piecewise linear functions
defined on rational triangulations of boxes with rational values in the vertices.
Then there is no algorithm $Q$ with the following specification:
\begin{itemize}
\item Q is quasi-decision procedure for the class of sentences
of the form $$\exists x\in [0,1]^d\;.\; f_1(x)=0\,\wedge\,\ldots\,\wedge\,f_n(x)=0\,\wedge \,g(x)\geq 0$$
where $(f,g): [0,1]^d\to \R^n\times \R$ and $d$ and $n$ are arbitrary.
\item Q can access all functions $f_j$, $g$ in the formula only via the oracle $\I(f_j)$, resp. $\I(g)$. 
That is, $Q$ can call $\I(f_j)$ and $\I(g)$ arbitrary many times but has no access to the syntactical representation of $f_j$ and $g$.
\end{itemize}
\end{theorem}
As will be seen from the proof in Section~\ref{sec:nonquasidec}, the second condition in Theorem~\ref{thm:nonquasidec} 
may be replaced by the alternative condition:
\begin{itemize}
\item Q does not terminate whenever the input is non-robust.
\end{itemize}
Whether or not the second condition in Theorem~\ref{thm:nonquasidec} can be omitted completely is---up to the best of our knowledge---an open problem.

\section{The Quasi-decision Procedure}
\label{problem}
In this section, we construct an algorithm that decides, whether a robust sentence in $\mathcal{B}$ is true. The algorithm serves purely for proving Theorem~\ref{thm:main}. We do not claim it to be practically efficient whatsoever and leave a practically efficient quasi-decision procedure for future work. 

For any formula $U\in\mathcal{B}$, variable $x$ and $x_0\in\R$ we denote by $\Subst{U}{x}{x_0}$ the formula derived from $U$ by substituting $x_0$ for $x$ in every free occurrence of $x$ in $U$. 
%Note that in cases where $x_0$ is not rational, strictly speaking, $\Subst{U}{x}{x_0}$ is not in the class $\mathcal{B}$ any more. However, for avoiding technical clutter, we do not explicitly formalize this.
We also allow $x$ to be an $n$-tuple of
variables, and $x_0\in\R^n$, in which case $\Subst{U}{x}{x_0}$ denotes the parallel substitution of entries of $x_0$ with their corresponding entries of $x$.  
% Strictly speaking, one cannot substitute constants into formulas: one would
% need something that Buchberger calls "name generator", assigning a name to
% every constant. This works even for uncountable domains! This results in
% uncountably many formulas with names, but the set of formulas without names
% is still countable.

% For simplicity, we will sometimes denote $S(p_0):\equiv\Subst{S}{p}{p_0}$
% the formula derived from $S$ by substituting $p_0\in P$ for the variable $p$.

In our algorithms, we use an alternative form of the Cartesian product
that concatenates tuples from the argument sets, instead of forming pairs. That
is, for sets $X\subseteq\mathbb{R}^n$ and $Y\subseteq\mathbb{R}^m$ it produces the set $\{
(x_1,\dots,x_n,y_1,\dots, y_n) \mid (x_1,\dots,x_n)\in X, (y_1,\dots,y_m)\in
Y\}$. Especially, for the set $\{ () \}$ containing the $0$-tuple, $\{ ()
\}\times X$ will be $X$. The width of $\{ () \}$, viewed as a box, is zero by definition.

%Similarly, for a term $t$ containing a free variable
%$x$, the term $t[x\leftarrow x_0]$ will denote the result of substituting $x_0$ for $x$.pa

We construct an auxiliary algorithm \emph{$\mathrm{CheckSat}(S,P,r)$} with the following specification:\pagebreak[3]
%\begin{specification}
%\label{spec:CheckSat}
\begin{description}
\item[Input: ] ~
  \begin{itemize}
  \item a formula $S$ from $\mathcal{B}$ in $l$ free variables $p$,
  \item an $l$-box $P$ bounding the free variables of $S$,
  \item $r\in\mathbb{Q}^{>0}$,
  \end{itemize}
 such that the width of $P$ is at most $r$.
\item[Output: ] a nonempty subset of $\{ \True, \False \}$
\end{description}
with the following two properties:
\begin{description}
\item[Correctness: ] If the algorithm returns $\{ \True \}$  ($\{ \False \}$), then for all $p_0\in P$, $\Subst{S}{p}{p_0}$ is robustly true (robustly false).
\item[Definiteness: ] If for a given $l$-box $P_0$ bounding the free variables of $S$, either for all $p_0\in P_0$ the sentence $\Subst{S}{p}{p_0}$ is
robustly true or for all $p_0\in P_0$ the sentence $\Subst{S}{p}{p_0}$ 
is robustly false, then there exists an $\varepsilon>0$
such that for every $r\leq \varepsilon$ and every sub-box $P\subseteq P_0$ with width smaller than $r$,
the algorithm returns $\{ \True \}$ or $\{ \False \}$ (as opposed to $\{
  \True, \False \}$).
\end{description}
%\end{specification}

CheckSat(S, P, r) terminates always, but may return the indefinite result $\{ \True, \False \}$. The existence of such an algorithm immediately implies Theorem~\ref{thm:main}, because then the algorithm below is a quasi-decision procedure for $\mathcal{B}$. 
\begin{ntabbing}
\hspace*{0.8cm}\= \hskip 0.8cm \= \hskip 0.8cm \= \hskip 0.8cm \= \hskip 0.8cm \=\kill
$\varepsilon\leftarrow 1$\\
\textbf{loop}\+\\
$R\leftarrow\mathrm{CheckSat}(S,\{ ()\},\varepsilon)$\hspace*{4cm}\\
\textbf{if} $|R|=1$ \textbf{then} \` // $R$ is either $\{\True\}$ or $\{ \False \}$\\
\> \textbf{return} $s$ s.t. $s\in R$\\
\textbf{else}\\
\> $\varepsilon\leftarrow\varepsilon/2$
\end{ntabbing}
Note that the specification of CheckSat does not only result in a quasi-decision procedure, but also checks robustness of the input. 

We will now define the algorithm $\mathrm{CheckSat}(S, P, r)$ in detail. We will leave the
proof that it fulfills the specification to Sections~\ref{sec:correctness}
(correctness) and~\ref{sec:definiteness} (definiteness).
The algorithm is recursive,  following the definition of class
$\mathcal{B}$. We will now describe the parts corresponding to the individual cases of this definition. 

%The following technical lemma, which holds by induction, using Items~(\ref{item:Bbasic})--(\ref{item:Bconn}) of
%Definition~\ref{classB}, ensures that we remain in the class $\mathcal{B}$:
%\begin{lemma}
%\label{evaluation}
%Let $U$ be a formula in $\mathcal{B}$ where $x$ occurs only freely and $x_0\in\R$. 
%Then $\Subst{U}{x}{x_0}$ is also in $\mathcal{B}$.
%\end{lemma}
%
%In all algorithms we work with intervals with rational endpoints. In particular,
%for any $\varepsilon>0$, we are able to split a given box $B$ into a grid of
%sub-boxes of width smaller then $\varepsilon$.
%

\subsection{System of Equations and Inequalities}
\label{SoE}

We first consider the case $(\ref{item:Bbasic})$ of class $\mathcal{B}$, that is, a formula $S$ of the form 
\[
\exists x\in B \;.\; [f_1=0\, \wedge\, \ldots\wedge\,
f_{n}=0\,\wedge\, g_1\geq 0\,\wedge\ldots\wedge\,
g_{k}\geq 0 ]
\]
where $B$ is an $m$-box. In an abuse of notation we also use $f_1,\dots,f_n$ and
$g_1,\dots,g_k$ for the functions denoted by those terms. 
They are functions in $P\times B\rightarrow \mathbb{R}$ with $P\times B\subseteq\mathbb{R}^{l+m}$, where $l$ is the
number of free variables of $S$. We assume that the order of the
arguments of those functions is the same as the order in which the respective variables
are quantified in the overall formula. Finally, we denote by $f:
P\times B\rightarrow \mathbb{R}^n$ the function defined by the components
$(f_1,\ldots, f_n)$ and by $g: P\times B\rightarrow \mathbb{R}^k$ the
function defined by the components $(g_1,\ldots,g_k)$. % Since the \st{composition} \Peter{(concatenation?)} of interval-computable functions is
% again interval-computable, $f$ and $g$ are interval-computable.

% Let $f=f(x,p)$ be a function with variables $x\in B$ and $p\in P$. 
% We adapt the convention that for $f(p_0)$
% refers to a function where $p$ is substituted by $p_0$. 

% Let us consider the formula
% $S\equiv(\exists x\in B\,\, (f(p_0)=0\,\wedge\,g(p_0)\geq 0))$ for all $p_0 \in P$. 
% %$\exists x\in B\,\, f_1(x)=0\,\wedge\, f_2(x)=0\,\wedge\,\ldots\,f_n(x)=0$ 
% of $n$ equations and $k$ inequalities in  $m$ variables defined on a  $B$, $m\leq n$.

Disproving the formula is straight-forward using the information given by $\I(f)$ and $\I(g)$. However, in order to ensure that the computed over-approximation is not too big, instead of working with $\I(f)(B)$ and $\I(g)(B)$ we work with elements of a partition $S_r$ of $B$ into small enough pieces, where ``small enough'' is determined by the parameter $r$ (Line~\ref{l:part} of the algorithm SoEI below). For this, we will call a set of boxes $S_r$ a \emph{grid covering $B$} iff  $\bigcup_{B'\in S_r} B'=B$ and for every $B^1\in S_r$ and $B^2\in S_r$, $\mathsf{int}(B^1)\cap \mathsf{int}(B^2)=\emptyset$.

The core of the algorithm for proving the formula is a test whether a system of
equations $f=0$ has a solution in a bounded region. The test analyzes the
boundary of the region and exploits continuity to deduce existence of a zero in
the interior.  

In the one-dimensional case, a bounded region is simply a closed
interval. If $f$ has opposite sign on the two end-points of the interval,
the intermediate value theorem tells us, that $f$ has a solution in the
interior. Here $f$ has to be non-zero on both interval endpoints (since $f$ is
in general non-polynomial, we cannot verify that $f$ is zero on an interval
endpoint, we can only exclude this). In general, we use the notion of the degree from the field of differential topology~\cite{Milnor:97,Outerelo:2009}.
For a continuous function $f:\Omega\to\R^n$ where $\Omega$ is a bounded open set and $p\notin f(\partial\Omega)$, the
degree of $f$ with respect to $\Omega$ and a point $p\in\R^n$ is an integer denoted by
$\deg(f,\Omega,p)$. If $\deg(f,\Omega,p)\neq 0$ then the equation $f=p$ has a
solution in $\Omega$. Since the degree is a non-trivial mathematical notion, we defer more details on the degree to Section~\ref{degdef} below.
%where we show that the degree is algorithmically computable for interval-computable analytic 
%functions $f: B\to\R^n$, with $B$ being an $n$-dimensional box.

For ensuring that the test $\deg(f,\Omega,p)\neq 0$ eventually succeeds we have to make sure that $\Omega$ encloses a robust zero closely enough (the notion ``closely enough'' will be made precise in Sections~\ref{technical_section} and~\ref{sec:definiteness}). So, also in this case, we work with the partition $S_r$ of $B$, and we compute the degree of
the individual pieces. However, for ensuring that $f$ is non-zero on the
boundary of the pieces, we merge those pieces of the partition $S_r$ for which we cannot prove that (Line~\ref{l:merge}).

Checking the inequalities is straight-forward (Lines~\ref{l:decomp}
to~\ref{l:ineq_end}) using $\I(g)$. In order to ensure that the used boxes are
small enough, we undo the mergings before the check (Line~\ref{l:decomp}) and
apply $\I(g)$ to the individual boxes (Line~\ref{l:ineq_iv}).

The algorithm looks as follows:\\\web{}{\pagebreak[3]}
\\
\textbf{Algorithm} $\mathrm{SoEI}(S,P,r)$ \hfill // System of equations and inequalities
\begin{ntabbing}
\hspace*{0.6cm}\= \hskip 0.6cm \= \hskip 0.6cm \= \hskip 0.6cm \= \hskip 0.6cm \=\kill 
Let $B$ be the $m$-box for the domain of the quantified variables in $S$.\label{}\\
Let $S_r$ be a grid of boxes covering $B$ s.t. \label{l:part}\\ \> each grid element has width at most $r$. \\
\textbf{if} for every box $A\in S_r$ \+\label{l:unsat} \\ 
%$\min_{j\in \{1,\ldots, k\}} \max \I(g_j)(A\times P)<0$ or $0\notin \I(f)(A\times P)$ \textbf{then}\+\\ 
%there is no $y_f\in \I(f)(P\times A),\, y_g\in \I(g)(P\times A)$ s.t. $y_f=0$ and $y_g\geq 0$  \textbf{then}\+\\ 
either $0\notin \I(f)(P\times A)$ or $\I(g)(P\times A)\cap [0,\infty)^k=\emptyset$ \textbf{then}\+\\
\textbf{return} $\{ \False \}$ $\quad$ \`// $f=0\,\wedge\,g\geq 0$ has no solution \-\-\label{} \\
\textbf{if} $m=n$ \textbf{then}\+\label{} \\
Merge all boxes in $S_r$ containing a common face $C$ s.t. $0\in \I(f)(P\times C)$.\label{l:merge} \\
Remove all grid elements in $S_r$ containing a face $C$ s.t. 
$C\subseteq\partial B$ and $0\in \I(f)(P\times C)$.\label{}\\
Let $p_0$ be an arbitrary element of $P$\label{}\\
\textbf{for} each grid element $A\in S_r$ \textbf{do}\+\label{}\\
\textbf{if} $\deg(f(p_0),A,0)\neq 0$ \textbf{then}  \` // equations hold, so
check inequalities \+\label{l:deg_test}\\ 
Let $S_r(A)$ be a grid of boxes covering $A$ of width at most $r$\label{l:decomp}\\ %// decompose what was merged\\
%{\bf if} for each $E\in S_r(A)$, $\min_{j\in\{1,\dots,k\}}\,\min \I(g_j)(E\times P)> 0$ \textbf{then}\+\\
{\bf if} for all $E\in S_r(A)$, $\I(g)(P\times E)\subseteq (0,\infty)^k$ \textbf{then}\+\label{l:ineq_iv}\\
  \textbf{return} $\{ \True \}$\-\-\-\-\label{l:ineq_end}\\
\textbf{return} $\{ \True, \False\}$ $\,\,$\label{}\`// no test succeeded, or
$n>m$ 
%\textbf{else} \+\label{}\\
%\textbf{return} $\{ \True, \False\}$\label{}
\end{ntabbing}
%
%\subsection{Inequalities}
%\label{SoI}
%Let us consider a more general formula of the form $(\ref{item:Bbasic})$ in Definition \ref{classB} corresponding to $u=1$. 
%In abbreviated form, we write it as 
%$\exists x\in B\,\, (f(p_0)=0\,\wedge\,g(p_0)\geq 0)$ for all $p_0 \in P$. 
%
%\bigskip
%
%{\it Algorithm $So\I(f,g,B,P,r)$ (System of inequalities and equalities):}
%\begin{tabbing}
%\hspace*{0.8cm}\= \hskip 0.8cm \= \hskip 0.8cm \= \hskip 0.8cm \= \hskip 0.8cm \= \kill
%Create a grid $S_r$ of boxes in $B$ of side-length at most $r$\\
%\textbf{if} for each box $A\in S_r$\+\\ 
%$\min_{j\in \{1,\ldots, k\}} \max \I(g_j)(A\times P)<0$ or $0\notin \I(f)(A\times
%P)$ \textbf{then}\+\\ 
%\textbf{return} $\{ \False \}$ \-\-\\
%Merge all boxes $A\in S_\varepsilon$ containing a face $C$ s.t. $0\in \I(f)(C\times P)$\\
%%Remove all grid elements $A$ containing a face $C\subset\partial B$ with $0\in \I(f)(C\times P)$\\
%\textbf{for} each grid element $A$\+\\
%\textbf{if} $SoE(f,A,P,r)=\{ T \}$ and $\min_j \min \I(g_j)(A\times P)>0$ \textbf{then} \+\\
%\textbf{return} $\{ \True \}$\-\-\\
%\textbf{return} $\{ \True, \False\}$
%\end{tabbing}
%%
Here we suppose that $f$ is present in the formula (i.e., $n>0$). The algorithm
can be easily adapted to the case, where it is not. In the case $n>m$, the
algorithm can simply return $\{ \True, \False\}$, see Lemma~\ref{m<n} below. An illustration of the algorithm
is shown in Figure~\ref{fig:SoEI}.

\begin{figure}[htb]
\begin{center}
  \includegraphics{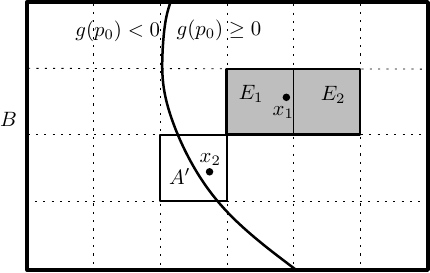}
  \caption{Illustration of the SoEI algorithm. Assume that  $f(p_0)$ has two zeros $x_1$ and $x_2$, and assume that $\{x\,|\,g(p_0,x)\geq 0\}$ is to the right of the thick curve. The algorithm creates a grid of boxes $S_r$ (line~\ref{l:part}). If each element of the grid provably does not contain a solution (check at line~\ref{l:unsat}), it returns $\{\False\}$. If this is not the case, then it checks whether $f$ is non-zero on all boundaries of grid elements~(line~\ref{l:merge}). In our example, $f$ is close to zero on the common boundary of $E$ and $E'$ and so the algorithm merges them into one grid element $A_1$.
  If $\deg(f(p_0), A_1, 0)\neq 0$, then it checks whether for each $p$, $g(p)\geq 0$ on
  $E_1$ and $E_2$ (line~\ref{l:ineq_iv}). If this is true as well,
  then $f=0\,\wedge\,g\geq 0$ is robustly satisfiable on $B$ and the algorithm terminates with $\{\True\}$. In case of another box $A_2$ containing a robust zero
  of $f$, the given partition may not provide enough evidence for the claim
  that $g(p, x_2)\geq 0$ for each $p$ (in which case the condition on line~\ref{l:ineq_iv} is not
  satisfied).}
\label{fig:SoEI}
\end{center}
\end{figure}

%The checks for
%positivity/negativity etc. of all box elements can be easily implemented by
%checking the end-points of the intervals defining the box.
%
%In the general case of a formula $(\ref{item:Bbasic})$ of Definition
%\ref{classB} ($u>1$), we denote by $f^1,\dots, f^u$ and $g^1,\dots, g^u$ the
%functions defined by the equalities (inequalities, respectively) of the
%individual disjunctive branches. We use the following algorithm.\\
%\\
%\textbf{Algorithm} $\mathrm{SoEI\_disj}(S,P, r)$:
%\begin{tabbing}
%\hspace*{0.8cm}\= \hskip 0.8cm \= \hskip 0.8cm \= \hskip 0.8cm \= \hskip 0.8cm \= \kill
%\textbf{return} $\{ \bigvee_i v_i \;|\; i\in \{ 1,\dots, u\}, v_i \in \mathrm{SoEI}(f^i,g^i,B,P,r)\}$
%\end{tabbing}

\subsection{Universal Quantifiers}
\label{Univ}
% Let us assume that $I$ is a bounded interval, $\forall x\in I\,\, S$ is a formula from $\mathcal{B}$. 
% %Let us denote by $S(x\lr x_0)$ or simply $S(x_0)$ the formula derived from $S$ by substituting $x_0\in\R$ for each occurance of $x$.
% Assume again that $P$ is a (possibly empty) domain of the free variables in $\forall x\, S$.
% %It follows from Lemma \ref{evaluation} that $S(x_0)$ is a formula in $\mathcal{B}$.

The recursive call corresponding to Case~(\ref{item:Buniv}) of class $\mathcal{B}$ looks as follows:\\
%Assume, for induction, that for any $x_0$, 
% we have an algorithm $Alg(S(x_0),P',r)$ ($r>0$, $P'\subset P$) fulfilling our specification.
%We propose the following algorithm for a sentence of the form $\forall x\in I\,\, S$:
\\
% \textbf{Algorithm} $\mathrm{Univ}(\forall x\in I\;.\; S,P,r)$:
% \begin{tabbing}
% \hspace*{0.8cm}\= \hskip 0.8cm \= \hskip 0.8cm \= \hskip 0.8cm \= \hskip 0.8cm \= \kill
% Let $I_r$ be a grid of sub-intervals $I$ of width at most $r$\\
% Let $(R_1,\dots, R_{|I_r|})$ be the $|I_r|$-tuple $(\mathrm{CheckSat}(S,P\times I',r)  \;|\; I'\in I_r)$\\
% \textbf{return} $\{ t_1\wedge\dots\wedge t_{|I_r|} \mid t_1\in R_1\dots t_{|I_r|}\in R_{|I_r|}\}$
% \end{tabbing}
\textbf{Algorithm} $\mathrm{Univ}(\forall x\in I\;.\; S,P,r)$:
\begin{tabbing}
\hspace*{0.8cm}\= \hskip 0.8cm \= \hskip 0.8cm \= \hskip 0.8cm \= \hskip 0.8cm \= \kill
Let $I_r$ be a grid of sub-intervals of $I$ of width at most $r$\\
\textbf{return} $\tilde{\bigwedge}_{ I'\in I_r} \mathrm{CheckSat}(S,P\times I',r)$
\end{tabbing}

Here, in the return statement, the symbol $\tilde{\bigwedge}$ denotes the lifting of
Boolean conjunction to sets of Boolean values: \[U \tilde{\wedge} V:= \{ u \wedge v\mid
u\in U, v\in V\}.\]

\subsection{Conjunctions and Disjunctions}
\label{cd}
% Let $S$ and $T$ be two formulas from class $\mathcal{B}$. 
% Let us assume that we are able to run $CheckSat(S,P_1,r)$ and $CheckSat(T,P_2,r)$ 
% for the formulas $S$ and $T$ containing free variables with domains $P_1$
% resp. $P_2$, fulfilling our specification.
% To check the validity of the sentence, $A\wedge B$, we have simply\\
Finally, the recursive call corresponding to Case~(\ref{item:Bconn}) of class $\mathcal{B}$ looks as follows:\\
\\
\textbf{Algorithm} $\mathrm{Conj}(S\wedge T,P,r)$
\begin{tabbing}
\hspace*{0.8cm}\= \hskip 0.8cm \= \hskip 0.8cm \= \hskip 0.8cm \= \hskip 0.8cm
\= \kill
\textbf{return} $\mathrm{CheckSat}(S,P_1,r) \:\tilde{\wedge}\: \mathrm{CheckSat}(T,P_2,r)$\\
\hspace*{4cm} where $P_1$ ($P_2$) is the projection of $P$\\
\hspace*{4.5cm} to the free variables of $S$ ($T$, respectively).\\ 
\end{tabbing}

Here, in the return statement, the symbol $\tilde{\wedge}$ again denotes the lifting of
conjunction to sets of Boolean values. The algorithm for disjunction is
completely analogous, replacing $\tilde{\wedge}$ with $\tilde{\vee}$ (and its lifting to sets of
Boolean values).

\section{Degree of a Continuous Function}
\label{degdef}

In this section we describe some basic properties of the topological degree.
%, and show, how it can be computed. 
We already mentioned in the introduction that in the one-dimensional case,
that is, for continuous functions $f: [a, b]\rightarrow \R$ with $f(a)\neq 0$ and $f(b)\neq 0$, the degree $\deg(f,
[a, b], 0)$ is $0$ iff $f(a)$ and $f(b)$ have the same sign, otherwise the
degree is either $-1$ or $1$, depending on whether the sign changes from
negative to positive or the other way round. Hence, in this case, the degree
gives the information given by the intermediate value theorem plus some
directional information.

In dimension two, the degree of a continuous function $f$ from a disc to
$\R^2$ is just the number of times $f(x)$ winds around the origin
counter-clockwise as $x$ follows
the circle forming the boundary of the disc (i.e., the ``winding
number''). Again, a non-zero winding number implies that $f$ has a zero. 

% For a more general continuous function
% between spaces of equal dimensions, the degree $\deg(f,U,a)$ is an integer
% defined in such a way that the equation $f(x)=a$ has a solution in $U$, if this
% degree is nonzero, and the degree does not change under small perturbations.
% This generalizes the fact that if a map from a circle to itself has nonzero winding number, 
% then the map is onto. 

There are several ways of defining the degree in general. We work with an axiomatic
definition, that can be shown to be unique~\cite[Section I.5]{Outerelo:2009}. Let $\Omega\subseteq\R^n$ be open and
bounded, $f:\bar\Omega\to\R^n$ continuous, and $p\notin
f(\partial\Omega)$. Then $\deg(f,\Omega,p)$ is an integer satisfying the
following properties~\cite[Thm. 1.2.6.]{Cho:06}:
\begin{enumerate}
\item  For the identity function $I$, $\deg(I,\Omega,p)=1$ iff $p\in\Omega$ 
\item\label{item:degsol}  If $\deg (f,\Omega, p)\neq 0$ then $f(x)=p$ has a
  solution in $\Omega$ 
\item\label{item:deghomot}  If there is a continuous function (a ``homotopy'') $h: [0,1]\times\bar\Omega\to\R^n$ such that
 $h(0)= f$, $h(1)=g$  and
$p\notin h(t,\partial\Omega)$ for all $t$, then
$\deg(f,\Omega,p)=\deg(g,\Omega,p)$
\item\label{item:union} If $\Omega_1\cap\Omega_2=\emptyset$, $\Omega_1\subseteq\Omega$, $\Omega_2\subseteq\Omega$, and $p\notin f(\bar\Omega\setminus(\Omega_1\cup\Omega_2))$,
then $\deg(f,\Omega,p)=\deg(f,\Omega_1,p)+\deg(f,\Omega_2,p)$
\item $\deg(f,\Omega,p)$, as a function of $p$, is constant on any connected
  component of $\R^n\backslash f(\partial\Omega)$.
\end{enumerate}

The first axiom says that for the identity function, the degree counts the zeros
in $\Omega$ precisely. Due to the second axiom one can infer existence of a zero
from a non-zero degree.  Due to the third axiom, the degree is invariant under
continuous deformations of the function that do not cause any essential change
of the boundary information. From this it can be immediately seen that the degree depends only on the boundary
$\partial\Omega$: for two functions $f$ and $g$ that agree on $\partial\Omega$,
the function $h(t,x)= t f(x) + (1-t) g(x)$ is a homotopy between $f$ and $g$, as
needed by the premise of Axiom~\ref{item:deghomot}.

% In the case where $f$ is continuous
% \marginpar{Remove this paragraph}
% and smooth (i.e., infinitely often differentiable) in $\Omega$ the degree can
% alternatively be defined based on the derivative of $f$. For regular
% values $p\in\R^n$ (i.e., values $p$ such that for all $y$ with $f(y)=p$, $\det
% f'(y)\not= 0$), a generalization of the directional information used in the
% one-dimensional case, is the sign of the determinant $\det
% f'(y)$. Adding up those signs results in the explicit definition~\cite{Milnor:97} of $\deg(f,\Omega,p)$ by
% $$\deg (f,\Omega,p):=\sum_{y\in f^{-1}(p)} \sign \det f'(y).$$ See standard textbooks for a generalization to non-regular values~\cite{Milnor:97}.
In the SoEI algorithm, we apply the degree to the triple $(f,A,0)$ where $A$ is not open but the closure of an open set (it is the union of boxes). 
For completeness, we define $\deg(f,A,p):=\deg(f,A^\circ,p)$ where $A^\circ$ is the interior of $A$, whenever $p\notin f(\partial A)$.

Many algorithms for computing the degree have been proposed~\cite{Erdelsky:73,Kearfott:79,Boult:86,Aberth:94,Franek:12b}. More specifically, if $B$ is an $n$-box, $f:B\to\R^n$ is interval computable, $0\notin f(\partial B)$ and an algorithm $\I(f)$ is given, then the degree $\deg(f,B,0)$ can be algorithmically computed. This justifies the use of
line 10 of algorithm SoEI in Section~\ref{SoE}.
%This can be extended to the case where $\Omega$ has dimension $n$ but is embedded into some
%higher-dimensional space  (in geometric terms, $f$ is a differentiable function
%between two compact oriented manifolds of the same dimensions).
%% such that $p\notin f(\partial\Omega)$).
%For example, if $f$ is a function from
%a segment $c$ of a curve (i.e., a set of dimension $1$) in $\R^k$ to another
%segment of a curve in $\R^k$, and if $f\neq 0$ on the endpoints of $c$, then
%$\deg(f,c,0)$ is well-defined. 
%% Following this approach\footnote{This does not
%%   connect any more to the text immediately before.}, we will write $\deg(f, B, 0)$, where $B$ will be a box
%% such that $f\neq 0$ on $\partial B$.
%

The axioms defining the degree only argue about zeros, but not about
robustness. Still, a nonzero degree is closely connected with the existence of a robust root:
\begin{lemma}
\label{degnonzero}
Let $\bar\Omega\subseteq\R^n$ be a closed region with interior $\Omega$,
$f:\bar{\Omega}\to\R^n$ be continuous, $0\notin f(\partial\Omega)$ and let 
$\deg(f,\Omega,0)\neq 0$. 

Then any continuous $g: \bar\Omega\to\R^n$ such that $\|g-f\|<\min_{x\in \partial\Omega} |f|$ has a zero in $\Omega$. 
\end{lemma}
\begin{proof}
Let $\varepsilon<\min_{x\in\partial\Omega} |f|$. For any $g$ such that $||g-f||_{\bar\Omega}<\varepsilon$, we define a homotopy
$h(t,x)=t f(x)+(1-t)g(x)$ between $f$ and $g$. We see that for $x\in\partial\Omega$ and $t\in[0,1]$,
$$
|h(t,x)|=|t f(x)+(1-t)g(x)|=|f(x)+(1-t)(g(x)-f(x))|\geq |f(x)|-\varepsilon > 0
$$
so that $h(t,x)\neq 0$ for $x\in\partial\Omega$.
From Properties~\ref{item:degsol} and~\ref{item:deghomot}, we see that $g(x)=0$ has a
solution. \qed
\end{proof}

In particular, this implies that the sentence $\exists x\in B\,\cdot\,f(x)=0$ is not only true, but also robust,  whenever $\deg(f, B^\circ, 0)\neq 0$. The upper bound on the distance between $f$ and $g$ results in an $\varepsilon$ such that this sentence is $\varepsilon$-robust. This allows extensions of the algorithms of this paper to return such an $\varepsilon$, which may be useful in applications.

For proving definiteness, we will need a partial converse of this statement which will be given by Theorem~\ref{robustzero} in Section~\ref{technical_section}.

\section{Proof of Correctness}
\label{sec:correctness}

We will prove here that the algorithm $\mathrm{CheckSat}$ proposed in Section~\ref{problem}
fulfills the first part of its specification, that is: it always returns a
correct result. The proof will again be divided into the cases constituting the
definition of class $\mathcal{B}$, from which correctness of the overall,
recursive algorithm follows by induction. 

Before that, we prove some technical results on the relationship between the class
$\mathcal{B}$ and robustness. 

Note that, in this section, the assumption that our symbol set contains addition and multiplication, is not used. Hence the algorithm is correct even if we do not have those symbols in the symbol set.

\subsection{Robustness and the Class $\mathcal{B}$}

First we prove a lemma on the effect of substitution of nearby constants on robustness.

\begin{lemma}
\label{lem:robust_neighborhood}
  Let $S$ be a formula in $l$ free variables, $P$ an $l$-box bounding the free variables of $S$ and $p_0$ be a point in the interior of $P$. If $\Subst{S}{p}{p_0}$ is a robust sentence, 
  then there exists a neighborhood $U\subseteq \R^l$ of
  $p_0$, such that for all $u\in U$, $\Subst{S}{p}{u}$ is robust and has the same truth value as $\Subst{S}{p}{p_0}$.
\end{lemma}

\begin{proof}
  Assume that $\Subst{S}{p}{p_0}$ is robust. Then there is an $\varepsilon>0$ such that for all formulas $T$ with $d(\Subst{S}{p}{p_0}, T)<\varepsilon$, $T$ and $\Subst{S}{p}{p_0}$ have the same truth value.
 Since all functions in $S$ are interval-computable, they are uniformly continuous.
%All quantified variables from $S$ are restricted to a compact box $B$ and each function occurring in $S$ is %continuous as a function $B\times P\to\R$: it follows that it is also uniformly continuous. 
Hence for $\varepsilon>0$, there exists a number $\delta>0$ such that for each function $f$ occurring in $S$ it holds that $|f(x,u)-f(x,p_0)|<\varepsilon/2$ whenever $u,p_0\in P$ and $|u-p_0|<\delta$. In other words, there exists a $\delta>0$ s.t. for all $u\in P$ with $|u-p_0|<\delta$, $d(\Subst{S}{p}{p_0}, \Subst{S}{p}{u})<\varepsilon/2$, and hence $\Subst{S}{p}{p_0}$ and $\Subst{S}{p}{u}$ have equal truth value. We claim that $\Subst{S}{p}{u}$ is 
also robust: this is because if $T'$ is any sentence with $d(\Subst{S}{p}{u}, T')<\epsilon/2$, then $d(T',\Subst{S}{p}{p_0})<\epsilon$ and
$T'$ has still the same truth value as $\Subst{S}{p}{p_0}$. So the neighborhood $U:=\{ u\in P \mid |u-p_0|<\delta \}$ of $p_0$ satisfies the required properties.  \qed
\end{proof}

  % Assume that $\Subst{S}{p}{p_0}$ is robust. Then there is an $\varepsilon>0$ such
  % for all $T$ with $d(\Subst{S}{p}{p_0}, T)<\varepsilon$, $T$ and
  % $\Subst{S}{p}{p_0}$ are equi-satisfiable. 
  % Since all quantified variables are from a compact box,
  % there exists a $\delta$ s.t. $|u-p_0|<\delta$ implies $d(\Subst{S}{p}{p_0}, \Subst{S}{p}{u})<\varepsilon$.
  % It follows that for $|u-p'|<\delta$,
  % also $\Subst{S}{p}{u}$ is equi-satisfiable and robust. \qed

Due to the syntactical structure of formulas in the class $\mathcal{B}$ we automatically have robustness in the false case:

\begin{lemma}
\label{robustfalse}
Let $S$ be a sentence from $\mathcal{B}$. If $S$ is false, then it is robustly false.
\end{lemma}
\begin{proof}
We proceed by induction, following the cases of class $\mathcal{B}$. 
Let $S$ be the sentence $\exists x\in B \;.\; f=0\,\wedge g\geq 0$, where $f=0$,
and $g\geq 0$ are the usual short-cuts for conjunctions of equalities, and
inequalities, respectively. Let $S$ be false. If $f=0$ has no solution in $B$, then
$||f||>\varepsilon$ for some $\varepsilon>0$ and $||\tilde{f}||>0$ for small enough perturbations $\tilde{f}$ of $f$.
Similarly, if $g<0$ on $B$, then the same is true for small enough perturbations of $g$.
Finally, if 
$f^{-1}(\{0\})\cap B$ and $g^{-1}[0,\infty)^k\cap B$ are both nonempty, then they are compact and disjoint, 
which implies that they have a positive distance. For small perturbations $\tilde{f}, \tilde{g}$
of $f$ and $g$, $\tilde{f}^{-1}\{0\}$ and $\tilde{g}^{-1}[0,\infty)^k$ are still disjoint, which implies
that $S$ is robustly false.

Further, assume that $I\subseteq\R$ is a compact interval and $\forall x\in I\;.\; S$ is a false sentence. 
Then there exists an $x_0\in I$ such that $\Subst{S}{x}{x_0}$ is false. From the induction hypothesis, it is robustly false. 
Let $\varepsilon>0$ be such that $\Subst{S}{x}{x_0}$ is $\varepsilon$-robust and let 
$S'$ be a formula such that
$d(\forall x\,.\,S',\,\forall x\,.\, S)<\varepsilon$. Then $d(\Subst{S'}{x}{x_0}, \Subst{S}{x}{x_0})<\varepsilon$
and $\Subst{S'}{x}{x_0}$ is false. So, $\forall x\in I\;.\; S'$ is false and it follows that
$\forall x\in I\;.\; S$ is robustly false. 

Finally, let $U$ and $V$ be sentences in $\mathcal{B}$ and $U\wedge V$ be false. Then either $U$ or $V$ is false
and the induction hypothesis says that it is robustly false. So, $U \wedge V$ is robustly false.
Similarly, if $U\vee V$ is false, then both $U$ and $V$ are robustly false and
$U\vee V$ is robustly false. \qed
\end{proof}

In the case of this lemma, the proof goes through for any number of equalities, independent of the restriction that class $\mathcal{B}$ puts on this number. Further, the last lemma remains true even if we leave the set of interval-computable functions and allow arbitrary, small enough continuous
perturbations. Moreover, it holds even if all functions in the original formula $S$ are only continuous and not interval computable.  We only have used continuity of the perturbations and the proof does not use any algorithmic input.

Universal quantification preserves robustness in the following sense:
\begin{lemma}
\label{robusttrue}
Let $S$ be a formula containing a free variable $x$ and let $I$ be a bounded closed interval. Then the sentence $\Subst{S}{x}{x_0}$ is robustly true for all $x_0$ in $I$ if and only if the sentence $\forall x\in I \;.\; S$ is robustly true.
\end{lemma}
\begin{proof}  
Let $\forall x\in I\;.\; S$ be $\varepsilon$-robust and true, and let $x_0$ be an arbitrary, but fixed element of the interval $I$. Then clearly $\Subst{S}{x}{x_0}$ is true. For showing that it is also robust, we assume an arbitrary, but fixed sentence $X$ such that
$d(X,\Subst{S}{x}{x_0})=:\varepsilon'<\varepsilon$ and prove that $X$ is true, as well. 
Let $f_X$, resp. $g_X$ be the functions that occur in $X$ on the places corresponding to $\Subst{S}{x}{x_0}$; this is well-defined, because $X$
and $\Subst{S}{x}{x_0}$ have the same structure.
Consider the formula $U$ that is equal to $S$ except for the fact that every equality of the form $f=0$ is replaced by $f+f_X-\Subst{f}{x}{x_0}=0$ and $g\geq 0$ is replaced by $g+g_X-\Subst{g}{x}{x_0}\geq 0$.
The distance $d(\forall x\in I\;.\; S, \forall x\in I\;.\; U)=\varepsilon'<\varepsilon$ and so, due to $\varepsilon$-robustness of $\forall x\in I\;.\; S$, $\forall x\in I\;.\;U$ is true. In particular, $\Subst{U}{x}{x_0}\equiv X$ is true and it follows that $\Subst{S}{x}{x_0}$ is $\varepsilon$-robust and true.

For the converse, assume that for all $x_0\in I$, $\Subst{S}{x}{x_0}$ is robustly true. Let 
$$\mu(x_0):=\sup \{\mu>0;\,\text{$\Subst{S}{x}{x_0}$ is $\mu$-robust}\}.$$
Clearly, $\mu$ is a continuous function in $x_0$ and has strict lower bound $m>0$ on the compact interval $I$.
So, for each $x_0\in I$, $\Subst{S}{x}{x_0}$ is $m$-robust. If $d(\forall x\in I\;.\; S, \forall x\in I\;.\; U)<m$,
then for each $x_0\in I$, $d(\Subst{S}{x}{x_0}, \Subst{U}{x}{x_0})<m$ and $\Subst{U}{x}{x_0}$ is true. So, $\forall x\in I\;.\; U$
is true and $\forall x\in I\;.\; S$ is robustly true. \qed
\end{proof}
Again, the last lemma remains true in the stronger formulation where we consider a statement robustly true iff any small enough continuous perturbation of its function symbols is true---that is, perturbation by functions that do not necessarily correspond to terms formed from the given set of function symbols or functions that are not necessarily interval computable.

\subsection{System of Equations and Inequalities}

For proving correctness of the algorithm $\mathrm{CheckSat}$ we again start with the case $(\ref{item:Bbasic})$ of class $\mathcal{B}$, that is, a formula $S$ of the form 
\[
\exists x\in B \;.\; [f_1=0\, \wedge\, f_2=0\,\wedge\, \ldots\wedge\,
f_{n}=0\,\wedge\, g_1\geq 0\,\wedge\, g_2\geq 0\,\wedge\ldots\wedge\,
g_{k}\geq 0 ]
\]
where $B$ is an $m$-box. Assuming that the formula has $l$ free variables, we
again denote by $f: \mathbb{R}^{l+m}\rightarrow \mathbb{R}^n$ the function defined by the
components $(f_1,\ldots, f_n)$ and $g:\mathbb{R}^{l+m }\rightarrow \mathbb{R}^k$ the function defined by the components
$(g_1,\ldots,g_k)$. 

\begin{theorem}
  The algorithm $\mathrm{SoEI}(S,P,r)$ fulfills the correctness property of  the specification of $\mathrm{CheckSat}(S, P, r)$ (defined at the beginning of Section~\ref{problem}).
\end{theorem}
\begin{proof}
 
Assume first that the algorithm terminates with a negative result
$\{\False\}$. It follows directly from Definition~\ref{def:intcomp}, that 
the input sentence $\Subst{S}{p}{p_0}$ is false for any $p_0\in P$. Lemma~\ref{robustfalse} implies robustness.

Now assume that it terminates with a positive result $\{ \True \}$. Then there exists
a point $p_0\in P\subseteq \mathbb{R}^l$ and a connected grid element $A\subseteq \mathbb{R}^{m}$ such that $\deg(f(p_0),A,0)\neq 0$.
For any $p\in P$, $p$ and
$p_0$ can be connected by a curve $\phi: [0,1]\rightarrow P$, and $f\circ \phi$ is then 
a~homotopy between $f(p_0)$ and $f(p)$ nowhere zero on $\partial A$.  So,
$\deg(f(p),A,0)\neq 0$ and it follows from Lemma~\ref{degnonzero} that
$f(p)=0$ has a robust solution in $A$.  Moreover,
the successful check whether for all $E\in S_r(A)$, $\I(g)(P\times E)\subseteq (0,\infty)^k$
implies that for some small enough $d>0$, for all $p\in P$, $x\in A$ and $j=1,\ldots, k$, $g_j(p, x)>d$. 
It follows that the input formula is robustly true for all parameter values in $P$. \qed
\end{proof}
\subsection{Universal Quantifiers}

\begin{theorem}
Let $S$ be a formula containing free variables $p$. Let $P$ be an $l$-box and $I$ a closed interval.
Assume that an algorithm $\mathrm{CheckSat}$ fulfilling the correctness property
is given. Then also the algorithm $\mathrm{Univ}(\forall x\!\in\! I \;.\;
S,P,r)$ fulfills the correctness property.
\end{theorem}
\begin{proof}
  If $\mathrm{Univ}(\forall x\!\in\! I \;.\; S,P,r)$
  returns $\{\False\}$, then $\mathrm{CheckSat}(S,P\times I',r')$ returned
   $\{\False\}$ for some  $I'\in I_r$ and it follows that for all $p_0\in P$ and
  $x_0\in I'$, $\Subst{\Subst{S}{p}{p_0}}{x}{x_0}$ is robustly false. Then
  $\forall x\in I \;.\; \Subst{S}{p}{p_0}$ is false for each $p_0\in P$ and it
  follows from Lemma~\ref{robustfalse} that it is robustly false.

If the algorithm returns $\{\True\}$, then $\mathrm{CheckSat}(S,P\times
I',r')$ returned $\{\True\}$ for all $I'\in I_r$ and the sentence
$\Subst{\Subst{S}{p}{p_0}}{x}{x_0}$ is robustly true for all $x_0\in I$ and
$p_0\in P$. It follows from Lemma~\ref{robusttrue} that for each $p_0\in P$,
$\forall x\in I\;.\; \Subst{S}{p}{p_0}$ is robustly true, so the result is
correct. \qed
\end{proof}

\subsection{Conjunction and Disjunction}

\begin{theorem}
  Let $S$ and $T$ be two formulas in $\mathcal{B}$  and assume that $\mathrm{CheckSat}$ fulfills the correctness property both when applied to $S$, and when applied to $T$. Then  $\mathrm{Conj}(S\wedge T, P, r)$ also fulfills the correctness property.
\end{theorem}

\begin{proof}
Let $p_S$, and $p_T$, respectively, be the function that projects any
$l$-tuple corresponding to the free variables of $S\wedge T$ to those
components corresponding to the free variables of $S$, and $T$, respectively.

If $\mathrm{Conj}$ returned $\{\True \}$ then the recursive calls for both $S$ and $T$ returned 
$\{ \True \}$. Hence, by correctness of the result of the recursive calls,
for all $p_0\in P$, $\Subst{S}{p_S(p)}{p_S(p_0)}$ and $\Subst{T}{p_T(p)}{p_T(p_0)}$ are robustly true, and 
hence also $\Subst{(S\wedge T)}{p}{p_0}$.

If $\mathrm{Conj}$ returned $\{ \False \}$ then the recursive calls for either
$S$ or $T$ returned $\{ \False \}$. Hence, by correctness of the result of the
recursive calls, either for all $p_0\in P$, $\Subst{S}{p_S(p)}{p_S(p_0)}$ is robustly
false, or for all $p_0\in P$, $\Subst{T}{p_T(p)}{p_T(p_0)}$ is robustly false. Hence, also
for all $p_0\in P$, $\Subst{(S\wedge T)}{p}{p_0}$ is robustly false. \qed
\end{proof}

For disjunctions the situation is analogous. 

\section{From Robustness To Non-Zero Degree}
\label{technical_section}

For proving that the algorithm CheckSat fulfills the second part of its specification, definiteness, we need to prove that for a robust system of equations, the test provided by a non-zero topological degree eventually succeeds. While the algorithmic aspects of the proof are part of the next section, in this section we prove two properties of the degree necessary for this (Lemma~\ref{m<n} and Theorem~\ref{robustzero}). The first property, Lemma~\ref{m<n}, simply says that in the case overdetermined system of $n$ equations in $m<n$ variables, the input cannot be robust, and hence the implication (robust input implies succeeding test for non-zero degree) holds
vacuously. The second property, Theorem~\ref{robustzero}, shows that robustness implies existence of
a region for which the degree is non-zero. More precisely, we will show a partial converse to Lemma~\ref{degnonzero}, that is, that a robust solution of $f=0$ on $\Omega$ implies the existence of a region $U\subseteq\Omega$ s.t. $0\notin f(\partial U)$ and $\deg(f, U,0)\neq 0$.

The rest of the paper will only refer to the two mentioned properties, so a reader can safely skip this section after noting Lemma~\ref{m<n} and Theorem~\ref{robustzero}. The proofs in the section are the only place in the
paper that uses results from topology that are not explicitly delineated in this paper. 

%In the case of a formula $S\equiv(\exists x\in B \,\, f(x)=0)$, where $f$ is a
%function from an $m-$box to $\R^n$, we say that {\it f has a robust zero in $B$}
%if the sentence $S$ is robustly true. First we prove that 
%a function from an $m$-dimensional box to a higher-dimensional space $\R^n$
%($n>m$) never has a robust zero. 
%Then  we prove that if $f: \Omega\subset\R^n\to\R^n$ and
%$\deg(f, \Omega, 0)\neq 0$, then f has a robust zero in $\Omega$.
%In the rest of the section we prove a partial converse of this. 
%We will show that if the degree is zero and the set of solution $f=0$ is connected, 
%then $f$ does {\it not} have a robust zero in $\Omega$. Finally, we will show that 
%if $f=0$ has a robust solution in $\Omega$, then there exists an open set $U\subset\Omega$
%such that $0\notin f(\partial U)$ and $\deg(f,U,0)\neq 0$.

%We start with an observation that overdetermined systems of equations never have a robust solution.
\begin{lemma}
\label{m<n}
Let $\bar\Omega$ be a closed region in $\R^m$, $n>m$ and $f:\Omega\to\R^n$ be continuous. Then for each $\epsilon>0$ there exists
a function $g: \bar\Omega\to\R^n$, $\|g-f\|<\epsilon$, with no root.
\end{lemma}
\begin{proof}
We assume that for some $\epsilon$, it holds that each $g$ closer to $f$ than $\epsilon$ has a root, and derive a contradiction. It follows from the Stone-Weierstrass theorem that the continuous function $f$ may be approximated arbitrarily precisely with a smooth function (even with a polynomial), 
and so we can approximate it by a smooth function $\tilde{f}$ closer than $\epsilon/2$ to $f$. Moreover, each such $\tilde{f}$ with $\|\tilde{f}-f\|<\epsilon/2$ has a root.
In particular, $\tilde{f}(x)-c$ has a root for any constant $c$, $|c|<\epsilon/2$ and so $\tilde{f}(\Omega)$ contains 
a~neighborhood of $0\in\R^n$.  However, all values in $\tilde{f}(\Omega)$ are critical values (that is, for each $x\in\bar\Omega$, 
the rank of $f'(x)$---a matrix $n\times m$---is smaller than $n$). Due to Sard's theorem~\cite[Chapter 2]{Milnor:97} 
the set of critical values of a smooth function has zero measure in $\R^n$, and so $\tilde{f}(\Omega)$ cannot contain 
a neighborhood of $0\in\R^n$, a contradiction. \qed
\end{proof}

%The same argument is valid if $\Omega$ is an $m$-dimensional manifold and $f:\Omega\to\R^n$.
%
%We start this section with the following basic property of the topological degree:
%
The rest of the section considers the case of equal dimensions $m=n$. First we
show that a zero degree of a function implies that any possible zero of the
function can be removed by a change of the function only in the
interior. Moreover, the result of the change will be small in a certain sense. 
\begin{lemma}
\label{degzerolemma}
Let $\bar\Omega$ be a closed region in $\R^n$, $f:\bar\Omega\to\R^n$ continuous,
$0\notin f(\partial\Omega)$ and
$\deg(f,\Omega, 0)=0$. Then there exists a continuous nowhere zero function $g:\bar\Omega\to\R^n$ 
such that $g=f$ on $\partial\Omega$ and $||g||_{\bar\Omega}\leq ||f||_{\bar\Omega}$.
\end{lemma}

\begin{proof}
If $0\notin f(\Omega)$, we may take $g=f$. Otherwise, take a neighborhood $U\subseteq\Omega$ of $f^{-1}(0)$
such that $\partial U$ is an $(n-1)$-manifold (i.e. locally homeomorphic to $\R^{n-1}$). Such a neighborhood
$U$ might be constructed as a finite union of balls. It follows from the degree axioms 
that $\deg(f,U,0)=\deg(f,\Omega,0)$ and it is a well-known fact in differential topology
that $f/|f|: \partial U\to S^{n-1}$ can be extended to a function $g_1: U\to S^{n-1}$ iff
the degree is zero~\cite[Theorem 8.1.]{Hirsch:76}. Let $h: \bar{U}\to \R^+$ be an extension of $|f|: \partial U\to\R^+$
(such extension exists due to Tietze's Extension Theorem~\cite[Thm. 4.22]{Bruckner:1997})
and let $i: S^{n-1}\to\R^n\setminus\{0\}$ be the inclusion.
Then $g_2:=h\,(i\circ g_1): \bar{U}\to\R^{n}\setminus\{0\}$ is a nowhere zero extension of $f|_{\partial U}$.
Define $g: \bar\Omega\to\R^n\setminus\{0\}$ by $g(x)=g_2(x)$ for $x\in U$ and $g(x)=f(x)$ for $x\notin U$.
This function is continuous, nowhere zero and coincides with $f$ on $\partial U$. Possibly multiplying 
$g$ by a positive scalar valued function that equals 1 on $\partial U$ and is small inside $U^\circ$, 
we achieve that $||g||_\Omega \leq ||f||_\Omega$.
\qed\end{proof}

\Long{Using the last theorem, we see that even if $0$ is in the interior of $f(\bar\Omega)$, $0$ does not have to be a robust zero of $f$. 
Let $\Omega$ be an open unit ball in $\R^2$. Let $(r,\varphi)$ be polar coordinates on $\Omega$, i.e. $0\leq r<1$ and
$\varphi\in\R$. Define the function $f:\Omega\to\R^2$ in polar coordinates by $f(r,\varphi)=(r, 4\sin\varphi)$.
This map preserves the diameter, and restricted to any circle of diameter $r$, the image is the whole circle,
because the angle $\varphi$ ranges from $-4$ to $4$, covering an interval larger then $2\pi$. 
So, the image of $f$ is the whole unit ball $\Omega$ and $0$ is in the interior of $f(\Omega)$. However,
the degree $\deg(f, B_r(0),0)$ is clearly zero for each ball $B_r(0)$ of radius $r$ (the boundary map from the unit circle to itself has zero 
winding number), so we may change $f$ inside $B_r$ to construct a nowhere zero $2r$-perturbation of $f$. 
This can be done for each $r$, so $0$ is not a robust zero of $f$.}

Now we show that for a smooth function $f$, we might change it within a small region $N$ where the function is nonzero, to produce
arbitrary many regular zero points, both orientation-preserving and orientation-reversing.
\begin{lemma}
\label{createzeros}
Let $U$ be an open set in $\R^n$, $f:U\to\R^n$ be smooth. Let $N$ be a neighborhood of $x^0\in U$ such that
$0\notin f(N)$ and let $k\in\N$.
Then there exists a function $f_1$ such that the following conditions are satisfied:\\
(1) $f_1=f$ on $U\setminus N$\\
(2) $||f_1||\leq ||f||$ \\
(3) $0$ is a regular value of $f_1|_N$\\
(4) $N$ contains $2k$ points $x_1,\ldots, x_k, y_1,\ldots, y_k$ such that
$f_1(x_i)=f_1(y_i)=0$, $f_1$ is orientation-preserving in the neighborhood of $x_i$ and orientation-reversing
in the neighborhood of $y_i$.
\end{lemma}
\begin{proof}
Choose $\delta>0$ such that $x^0+[-2\delta,2\delta]^n\subseteq N$.
We construct $f_1$ such that $f_1(x)=f(x)$ for $x \notin
(x^0+[-2\delta,2\delta]^n)$.  For $x \in (x^0+[-\delta,\delta]^n)$ we
set 
$$(f_1)_i(x)=\left(\frac{|x_i - x^0_i|}{\delta}
- \frac{1}{2}\right) f_i(x^0).
$$

It is easy to see that $f_1^{-1}(0)$ contains in $(x^0+[\pm\delta])$ $2^n$ points
of the form $(x^0_1 \pm \delta/2, x^0_2 \pm \delta/2,\ldots, x^0_n\pm\delta/2)$,
half of them preserve orientation and half reverse orientation. Clearly, 
$|f_1(x)|\leq |f(x^0)|\leq ||f||$ on $x_0+[\pm\delta]$. Because 
$\deg(f_1,x^0+[\pm\delta],0)=\deg(f,x^0+[\pm 2\delta])=0$, it is easy to see that
$f_1$ may be extended to $x^0+[\pm 2\delta]$ so that $f_1=f$ on $\partial (x^0+[\pm 2\delta])$,
$f_1$ is nonzero in $x_0+[\pm 2\delta]\setminus (x_0+[\pm\delta])$
and the norm $||f_1||\leq ||f||$. The only zero points of $f_1$ in $N$ 
are $(x^0_1\pm\delta/2,\ldots,x^0_n\pm\delta/2)$,
so  $0$ is a regular value of $f_1|_N$. The details are left to the reader.

To produce more zeros we can choose any point $x_1\in N$ s.t. $f_1(x_1)\neq 0$ and a small neighborhood of $x_1$
in $N$ where $f_1$ is nonzero and continue in the same way.
\qed\end{proof}

Finally, we prove the following theorem that will be used in the proof of definiteness of the CheckSat procedure.
\begin{theorem}
\label{robustzero}
Let $\bar\Omega$ be a non-empty closed region in $\R^n$ with interior $\Omega\subseteq\R^n$ and $f:\bar\Omega\to\R^n$ be continuous. 
Then there exists an $\epsilon>0$ such that each continuous $g$, $\|g-f\|<\epsilon$, has a zero in $\bar\Omega$
\emph{if and only if} 
there exists an open set $U\subseteq\Omega$ such that $0\notin f(\partial U)$  and  $\deg(f,U,0)\neq 0$.
\end{theorem}

The assumption that $\Omega$ is the interior of $\bar\Omega$ is necessary
to exclude some degenerate cases such as $\Omega=(-1,1)\setminus\{0\}$
and $f(x)=x$; in this case, $f$ has a robust zero in $\bar\Omega=[-1,1]$ 
but for any $U\subseteq\Omega$
with $0\notin f(\partial U)$, $\deg(f,U,0)=0$.

%In the proof, we will call a continuous function $g$ with $\|g-f\|<\epsilon$ a continuous $\epsilon$-perturbation %of $f$.

\begin{proof}
If the dimension is $n=1$, then $\bar\Omega$ is a compact interval and clearly there exists an $\epsilon>0$ such that each 
continuous $\epsilon$-perturbation of $f$ has a zero iff there exists $x,y\in\bar\Omega$
s.t. $f(x)<0<f(y)$, and the statement follows. In the rest of the proof we
assume that $n\geq 2$.

If $\deg(f,U,0)\neq 0$ for some $U$, then we may choose $\epsilon:=\min_{x\in\partial U} |f|$ by
Lemma~\ref{degnonzero} which proves one implication. 

For proving the other direction, we assume that for each open $U\subseteq\Omega$ s.t. $0\notin f(\partial U)$,
$\deg(f,U,0)=0$. We choose a positive $\epsilon>0$ and will show that there exists a continuous $4\epsilon$-perturbation $g$ 
of $f$ with no root.

Let $\Omega_\epsilon:=\{x\,|\:|f(x)|<\epsilon\}$.
This is an open set in $\bar\Omega$. 
Let $x\in f^{-1}(0)\cap\Omega$. Then there exists a 
ball $U(x)\subseteq\R^n$ open in $\R^n$ such that $U(x)\subseteq\Omega_{\epsilon}$. For $y\in f^{-1}(0)\cap\partial\Omega$,
we choose $U(y)\subseteq\R^n$ to be an open ball in $\R^n$ such that $U(y)\cap\bar\Omega\subseteq\Omega_{\epsilon}$.
We assumed that $\Omega$ is the interior of $\bar\Omega$, which implies $\partial\Omega=\partial\bar\Omega$. So, for each such
$U(y)$, the set $U(y)\setminus\bar\Omega$ is a nonempty open set in $\R^n$.

The set $\{U(x)\,|\,x\in f^{-1}(0)\}$ is an open cover of the compact set $f^{-1}(0)$, so we may take
finitely many of these sets $U_1,\ldots, U_k$ that still cover  $f^{-1}(0)$.
Each $U_i$ is either contained in $\Omega_{\epsilon}$, or has a nontrivial intersection with $\partial\Omega$.
Let $V_1,\ldots, V_l, W_1,\ldots, W_m$ be the pairwise disjoint connected components of $\cup_i U_i$ such that
$V_i\subseteq \Omega_{\epsilon}$ and $W_j\cap\partial\Omega\neq\emptyset$ for each $i,j$.

If $x\in \partial V_i$, then $f(x)\neq 0$, otherwise $x$ would be contained in the interior of the same connected component $V_i$ of $\cup_i U_i$. 
In particular, $0\notin f(\partial V_i)$ and
due to the assumption above $\deg(f,V_i,0)=0$. $V_i$ is connected and it follows from Lemma~\ref{degzerolemma} that we may change $f$ inside $V_i$,
without changing it on $\bar\Omega\setminus V_i$, to construct a function $f_1: \bar\Omega\to\R^n$, $0\notin f_1(V_i)$ and 
$||f_1||_{V_i}\leq ||f||_{V_i}$. The inequalities $||f_1||_{V_i}\leq ||f||_{V_i}\leq\epsilon$ imply that $f_1$ is a continuous $2\epsilon$-perturbation of $f$.
This can be done independently for each $i$, so we may assume that $0\notin f_1(\cup_i V_i)$.

Let us extend $f_1$ to a continuous function $f_2: \bar\Omega\cup_j \bar{W}_j
\to\R^n$ (such an extension exists by Tietze's Theorem).
Possibly multiplying $f_2$ by a positive scalar valued function that equals $1$ on $\bar\Omega$ and is small outside $\bar\Omega$, 
we may assume that $||f_2||_{\cup_j W_j}\leq \epsilon$.
The zero set of $f_2$ is contained in $\cup_j \bar{W_j}$ and if $f_2(x)=0$ for some $x\in\partial W_j$, then $x\notin\bar\Omega$
(otherwise, $x$ would be contained in the same connected component of $\cup_i U_i$ as $W_j$, contradicting $x\in\partial W_j$). 
Therefore, $f_2$ is nowhere zero on the compact set $\bar\Omega\setminus\cup_j W_j$ 
and there exists some $0<\epsilon_1<\epsilon$ s.t. $|f(x)|>\epsilon_1$ for $x\in\bar\Omega\setminus\cup_j W_j$.
Let $f_3$ be a continuous $\epsilon_1$-perturbation of $f_2$ that is smooth and $0$ is a regular value of $f_3$ 
(such a perturbation exists by Stone-Weierstrass and Sard's theorems). The set $f_3^{-1}(0)$ is finite and contained in $\cup_j \bar{W_j}$. 
For each $j$ and each $x\in f_3^{-1}(0)\cap \partial W_j$, we may find a small neighborhood $O_x$ of $x$ such that $x$
is the only zero point of $f_3$ on $\bar O_x$, $\bar O_x\cap\bar\Omega=\emptyset$, $W_j\setminus \bar O_x$ is still connected, 
and replace $W_j$ by $W_j\setminus \bar O_x$. 
So, we can assume that $0\notin f_3(\partial W_j)$
for each $j$. Let $A^+(W_j)=\{x\in W_j\,|\,f_3(x)=0,\,\det(f_3'(x))>0\}$ and $A^-(W_j)=\{x\in W_j\,|\,f_3(x)=0,\,\det(f_3'(x))<0\}$.

$W_j\setminus\bar\Omega$ is open and nonempty, and we can use Lemma~\ref{createzeros}
to create at least $2||A^+(W_j)|-|A^-(W_j)||$ zeros in $W_j\setminus\bar\Omega$ of $f_3$ in which $f_3$ is orientation-preserving, resp. orientation-reversing,
without changing $f_3$ in $W_j\cap\bar\Omega$. We can then pair all points in $A^+(W_j)\cap\bar\Omega$ with 
points in $A^-(W_j)\setminus\bar\Omega$ and points in $A^-(W_j)\cap\bar\Omega$ with $A^+(W_j)\setminus\bar\Omega$ (some zeros of $f_3$
outside $\bar\Omega$ may still remain unpaired). We suppose that the dimension $n\geq 2$, so we may connect each pair of points $x_a^+$ and
$x_a^-$ by a curve $c_a$ so that the curves do not intersect themselves and the complement of these curves in $W_j$ is still connected. 
Further, there exist connected and pairwise disjoint open neighborhoods $N_a$ of these curves such that the only zero points of 
$f_3$ in $N_a$ are $x_a^+$ and $x_a^-$ for each $a$.
The degree $\deg(f_3, N_a, 0)=0$, so we may change $f_3$ inside $N_a$ to a continuous function $f_4$ s.t. $||f_4||_{N_a}\leq ||f_3||_{N_a}$,
and $0\notin f_4(N_a)$. In this way, we destroy all zeros of $f_3$ in $\bar\Omega$ (although some zeros may still exists outside $\bar\Omega$).
We assumed that $||f_2||_{W_j}\leq \epsilon$, so $||f_3||_{W_j}\leq \epsilon+\epsilon_1\leq 2\epsilon$ and
$f_4|_{W_j}$ is a continuous $4\epsilon$-perturbation of $f|_{W_j}$. Changing $f_3$ independently in each $N_a$,
the resulting function $f_4|_{\bar\Omega}$ is a nowhere zero continuous $4\epsilon$-perturbation of $f$.
\qed
\end{proof}

\Long{
As an illustration of the last theorem, we present an example where the image of $f$ contains a neighborhood of zero, although arbitrary close
to $f$ are continuous functions with no root. Let $\Omega$ be a unit ball and $f$ is given in polar coordinates by 
$f(r, \varphi):=(r, 4\cos \varphi)$.
This maps circles of radius $r$ to circles of radius $r$, each of them with winding number zero 
(it goes around the circle in one direction and then back). Also, the image of each circle is the full circle, because $4\cos \varphi$ ranges
from $-4$ to $4$, which is an interval of length larger than $2\pi$. So, the image of $f$ is the full unit ball $\Omega$, 
but $\deg(f,U,0)=0$ for any $U\subseteq\Omega$. 
It follows from the last theorem that arbitrary close to $f$ are continuous functions with no root.}

\section{Proof of Definiteness}
\label{sec:definiteness}

We will prove here that the algorithm $\mathrm{CheckSat}$ proposed in
Section~\ref{problem} fulfills the second part of its specification, that is,
definiteness. This will complete the proof of Theorem~\ref{thm:main}. The definiteness proof will again be divided into the cases constituting
the definition of class $\mathcal{B}$, from which correctness of the overall,
recursive algorithm follows by induction.

Unlike in Section~\ref{sec:correctness}, in this section, the assumption that the symbol set of our language contains rational constants, addition, and multiplication, and consequently all polynomials with rational coefficients is needed: it will allow us to construct terms representing functions that are arbitrarily close to a given continuous function.

\subsection{System of Equations and Inequalities}

We again start with the case $(\ref{item:Bbasic})$ of class $\mathcal{B}$, that is, a formula $S$ of the form 
\[
\exists x\in B \;.\; [f_1=0\, \wedge\, f_2=0\,\wedge\, \ldots\wedge\,
f_{n}=0\,\wedge\, g_1\geq 0\,\wedge\, g_2\geq 0\,\wedge\ldots\wedge\,
g_{k}\geq 0 ]
\]
where $B$ is an $m$-box. Assuming that the formula has $l$ free variables, we
again denote by $f: \mathbb{R}^{l+m}\rightarrow \mathbb{R}^n$ the function defined by the
components $(f_1,\ldots, f_n)$ and $g:\mathbb{R}^{l+m }\rightarrow \mathbb{R}^k$ the function defined by the components
$(g_1,\ldots,g_k)$. 

\begin{theorem}
The algorithm $\mathrm{SoEI}(S,P,r)$ described in Section~\ref{SoE} fulfills the
definiteness property of the specification of $\mathrm{CheckSat}$ (defined at the beginning of Section~\ref{problem}).
\end{theorem}
\begin{proof}

Let $P_0$ be an $l$-box bounding the free variables of $S$. We divide the proof into two parts:

\noindent {\bf Negative case: }

\nopagebreak[4] Assume that $\exists x\in B \;.\; f(p_0)=0\,\wedge\,g(p_0)\geq 0$ is robustly false for each $p_0\in P_0$. 
We construct an $\varepsilon>0$ such that for every $r\leq \varepsilon$ and every sub-box $P\subseteq P_0$ with width smaller than $r$,
the algorithm returns $\{ \False \}$:

The sets $X=\{(p,x)\in P_0\times B \mid f(p, x)=0\}$ and $Y=\{(p,x)\in P_0\times B
\mid g(p,x)\geq 0\}$ 
are compact and disjoint, so they have a positive distance. For a small enough $\alpha>0$, the sets 
$X'=\{(p,x)\in P_0\times B \mid |f(p, x)|\leq\alpha\}$
and $Y'=\{(p,x)\in P_0\times B \mid  g(p,x)\geq (-\alpha,\ldots,-\alpha)\}$ 
are still disjoint and have a positive distance $d>0$.\footnote{This follows from the fact that $X$ resp. $Y$ can be separated by
open $\epsilon'$-neighborhoods $U(X)$ resp. $U(Y)$ with positive distance from each other, and the fact that using the uniform continuity of $|f|$ and $g$,
$X'\subseteq U(X)$ and $Y'\subseteq U(Y)$ for $\alpha$ small enough.}
If $\varepsilon_0$ is small enough, any box of width smaller than $\varepsilon_0$ either has
 an empty intersection with $X'$ or an empty intersection with $Y'$. 

The second property of interval computability implies that for $\alpha$ there exists a $\delta>0$ such that
any box $A\subseteq B$ with $\Width{A}<\delta$ and box $P\subseteq P_0$ with $\Width{P}<\delta$ have the following properties:
\begin{itemize}
\item If $P\times A$ has empty intersection with $X'$, then $0\notin \I(f)(P\times A)$.
\item If $P\times A$ has empty intersection with $Y'$, then $\I(g)(P\times A)\cap [0,\infty)^k=\emptyset$.
\end{itemize}

So, if we call the CheckSat algorithm with $r\leq \varepsilon:=\min\{\delta,\varepsilon_0\}$ and $P\subseteq P_0$ of width smaller than $r$, then for every $A\subseteq B$ in the resulting $S_r$-grid, 
either $P\times A$ has empty intersection with $X'$ or it has empty intersection with $Y'$ and due to the 
above properties, $A$ satisfies that $0\notin \I(f)(P\times A)$ or $\I(g)(P\times A)\cap [0,\infty)^k=\emptyset$. So the test at Line~\ref{l:unsat} of the algorithm succeeds and the algorithm terminates with $\{ \False \}$.\\

\noindent {\bf Positive Case: } 

\nopagebreak[4] Assume now that $\exists x\in B \;.\; f(p_0)=0\,\wedge\,g(p_0)\geq 0$ is robustly true for each $p_0\in P_0$. We prove that there exists an $\varepsilon>0$
such that for every $r\leq \varepsilon$ and every sub-box $P\subseteq P_0$ with width smaller than $r$,
the algorithm returns $\{ \True \}$.

Exploiting that our given set of functions symbols allows us to form polynomials with rational coefficients, it follows that for some $\alpha>0$, each 
$\alpha$-perturbation $\tilde{f}$ of $f(p_0)$  and $\tilde{g}$ of $g(p_0)$ such that each component of $\tilde{f}$ and of $\tilde{g}$ is a polynomial with rational coefficients,
 satisfies that 
$\exists x\in B \;.\; \tilde{f}=0\,\wedge\,\tilde{g}\geq 0$ is true. In particular, 
each polynomial $\alpha$-perturbation of $f(p_0)=0$ with rational coefficients has a root in 
the compact set $C:=\{x\in B \mid g(p_0,x)\geq \alpha\}$. 

Now we show that $m=n$. Otherwise $m<n$ and by Lemma~\ref{m<n} there exist arbitrary close continuous perturbations $\tilde{f}$ 
of ${f}(p_0)$ with no root in $C$. The absolute value of each such $\tilde{f}$ has 
a positive minimum on $C$ and arbitrary close to $\tilde{f}$ are rational polynomials with no root in $C$. But then arbitrary close to $f(p_0)$
would be polynomials with no root in $C$  which contradicts our assumption. 
Therefore, $m=n$.

We will now prove that for all $p_0\in P_0$ there is an open neighborhood $U(p_0)$ of $p_0$ and $\varepsilon(p_0)>0$ such that for all $P'\subseteq U(p_0)$, $\mathrm{SoEI}(S,P',\varepsilon(p_0))$ terminates with $\{\True\}$. So let $p_0\in P_0$ be arbitrary, but fixed, for which we will now construct such a $U(p_0)$ and $\varepsilon(p_0)$.

Let $\Omega_1\subseteq B$ be an open neighborhood of $C$ in $B$
such that $\Omega_1\subseteq \{x\in B \mid g(p_0,x)\geq\alpha/2\}$ and let $\Omega$ be the interior of $\bar\Omega_1$ in $\R^n$.
We already know that each small enough polynomial perturbation of $f(p_0)$ has a zero in $\bar\Omega$.

By construction, $\bar\Omega=\bar\Omega_1$ and $\bar\Omega$ is the closure of its interior, so we are now ready to use Theorem~\ref{robustzero}. 
It implies that there exists an open $U\subseteq\Omega$ such that $0\notin f(p_0)(\partial U)$ and $\deg(f(p_0),U,0)\neq 0$. Otherwise, by Theorem~\ref{robustzero} there would exist continuous perturbations of $f(p_0)$ with no zero in $\bar\Omega$ arbitrary
close to $f(p_0)$ and it easily follows that there would also exist rational polynomial perturbations arbitrary close to $f(p_0)$ with no zero in $\bar\Omega$.

While $\deg(f(p_0),U,0)\neq 0$ and the inequalities of $S$ strictly hold for all elements of $\{p_0\}\times U$, the set $U$ is not a union of boxes, and hence the algorithm will, in general, not come up with this set. So our goal is now to   construct $U(p_0)$ and $\varepsilon(p_0)$ in such a way that for all $P'\subseteq U(p_0)$, $\mathrm{SoEI}(S,P',\varepsilon(p_0))$ approximates $U$ closely enough for the degree test (Line~\ref{l:deg_test} of the algorithm) and the test of inequality satisfaction (Line~\ref{l:ineq_iv}) to succeed.

%We may assume (possibly taking a smaller $\alpha$) that $U\cap\partial\Omega=\emptyset$. (I don't need it)
Let $U(p_0)\subseteq P_0$ be an open neighborhood of $\{p_0\}$ in $P_0$ such that
$(U(p_0)\times\Omega)\subseteq \{(p,x) \mid g(p,x)\geq\alpha/4\}$\footnote{The set $\{(p,x) \mid g(p,x)>\alpha/4\}$ is an open neighborhood
of $\{p_0\}\times\bar\Omega$ and the compactness of $\bar\Omega$ implies that there is a neighborhood $U(p_0)$ of $\{p_0\}$ such that $U(p_0)\times \Omega\subseteq\{(p,x) \mid g(p,x)>\alpha/4\}$.}
and let $\varepsilon_g(p_0)$ be so small
that for every box $K\subseteq U(p_0)\times \Omega$ of width less than $\varepsilon_g(p_0)$,
\begin{equation}
\label{g-ineq}
\I(g)(K)\subseteq (0,\infty)^k
\end{equation}
which exists due to the second property of the definition of interval computability.

%Let $\varepsilon_1(p_0)<\varepsilon_0(p_0)$ be so that any $\varepsilon_1(p_0)$-box that meets $U$ does not meet 
%$\{f(p_0)=0\}-U$. Let $S_{\varepsilon_1(p_0)}$ be the grid of these boxes and merge all boxes containing
%a face $C$ such that $0\in \I(f(p_0))(C)$. Let $M$ be the merging of all $S_{\varepsilon_1(p_0)}$-grid element
%having nonempty intersection with $U$. Then $M$ contains $U$, $M\cap \{f(p_0)=0\}\subset U$ 
%and $\deg(f(p_0),M,0)=\deg(f(p_0),U,0)\neq 0$. It follows that there exists a grid element $A$ 
%(merging of $\varepsilon_1(p_0)$-boxes) such that $0\notin f(\partial A)$ and 
%$\deg(f(p_0),A,0)\neq 0$ (otherwise, $\deg(f(p_0),U,0)$ would be a sum of
%some $\deg(f(p_0),B_i,0)=0$ for $S_{\varepsilon_1}$-grid elements $B_i$ and therefore zero).

Possibly making $U(p_0)$ smaller, we may assume that $0\notin f(\overline{U(p_0)} \times\partial U)$. 
Let $V\subseteq\Omega$ be a neighborhood of $\partial U$ open in $B$ such that 
$0\notin f(\overline{U(p_0)}\times\bar{V})$ (in these constructions we exploit the compactness of $\partial U$, resp. $\overline{U(p_0)}$).
We will further assume that $U(p_0)$ is connected (if it were not, we could replace it by the connected component of $p_0$ in $U(p_0)$).
The compactness
of $\overline{U(p_0})\times\bar{V}$ implies that $|f|$ has a positive minimum on this set and 
the second property of the definition of interval computability
implies that there exists an $\varepsilon_f(p_0)$ such that for
every sub-box 
$K\subseteq U(p_0)\times V$
of width smaller than $\varepsilon_f(p_0)$, 
\begin{equation}
\label{I(f)}
0\notin \I(f)(K).
\end{equation} 
Let $\varepsilon_V(p_0)$ be such that each box of width less than $\varepsilon_V(p_0)$ that has a nonempty intersection with $\partial U$ lies in $V$. Let $\varepsilon(p_0)$ be $\min \{ \varepsilon_f(p_0), \varepsilon_g(p_0), \varepsilon_V(p_0) \}$.

Having constructed $U(p_0)$ and $\varepsilon(p_0)$ we will now show that they are indeed small enough for the algorithm to return a positive result: Let $P'\subseteq U(p_0)$ be a box of width at most $\varepsilon(p_0)$. 
We will show that $\mathrm{SoEI}(S,P',\varepsilon(p_0))$ terminates with $\{\True\}$.
The algorithm creates a grid of boxes $S_r$ such that each grid element has width at most $\varepsilon(p_0)$.
It merges boxes containing a face $C$ such that 
$0\in \I(f)(P'\times C)$ and removes elements (i.e. merged boxes) containing a face $C\subseteq\partial B$ such that $0\in \I(f)(P'\times C)$. Let us denote by $S_r^{\mathrm{m}}$ the set containing all these merged boxes after the removal. So, elements of $S_r^\mathrm{m}$ can be identified with unions of boxes in $S_r$. Let $M$ be the smallest union of elements in $S_r^\mathrm{m}$ such that $M\supseteq U$. $M$ consists of unions of boxes in $S_r$ that are either contained in $U$ or
intersect $\partial U$ and hence are contained in $V$. It follows that $M\subseteq\Omega$ (by a slight abuse of notation, we denote by $M$ both
the set of elements as well as the underlying space).
Further, $\partial M\subseteq V$, $0\notin \I(f)(P'\times C)$ for any boundary box $C\subseteq \partial M$
(due to $(\ref{I(f)})$)
and $$\deg(f(p'),M^\circ,0)=\deg(f(p_0), M^\circ, 0)=\deg(f(p_0), U, 0)\neq 0$$ for any $p'\in P'$. 
The first identity follows from the fact that $U(p_0)$ is connected, hence $p_0$ and $p'$ can be connected by a curve 
that gives rise to a homotopy between $f(p')$ and $f(p_0)$ that is nowhere zero on the boundary faces of $M$, see axiom~\ref{item:deghomot} defining the degree in Section~\ref{degdef}. 
The second identity follows from the fact that $0\notin f(M\setminus U)$ and axiom~\ref{item:union} of Section~\ref{degdef} applied to $\bar\Omega=M$, $\Omega_1=U$ and $\Omega_2=\emptyset$.

Let $p'\in P'$ be chosen in the algorithm. There exists a subset $M'\subseteq M$ that consists of elements in
$S_r$ where the algorithm finds that $\deg(f(p'),(M')^\circ, 0)\neq 0$ (otherwise, $M$ would be 
a union of subsets on which $f(p')$ has zero degree, contradicting $\deg(f(p'),M^\circ ,0)\neq 0$).
Then it splits elements of $M'$ back to the corresponding elements in $S_r$ and checks the condition whether
for all boxes $E\in S_r(M')$, $\I(g)(P'\times E)\subseteq (0,\infty)^k$. 
This is satisfied due to $(\ref{g-ineq})$ and the algorithm terminates
with $\{\True\}$.

So we now know that for all $p_0\in P_0$, there is an $U(p_0)$ and $\varepsilon(p_0)>0$ such that for all $P'\subseteq U(p_0)$, $\mathrm{SoEI}(S,P',\varepsilon(p_0))$ terminates with $\{\True\}$. So, we have a covering $\{U(p_0)\,|\,p_0\in P_0\}$ of the compact
set $P_0$ and can choose a finite sub-covering $\{U(p_1),\ldots,U(p_s)\}$. 
There exists an $\varepsilon'$ such that each box $P\subseteq P_0$ of width smaller than $\varepsilon'$
is contained in some $U(p_j)$. Let $\varepsilon$ be the minimum of $\varepsilon'$ and all the $\varepsilon(p_j)$, $j\in \{ 1,\dots,s\}$.
For any $P\subseteq P_0$ 
of width at most $\varepsilon$, $\mathrm{SoEI}(S,P,\varepsilon)$ terminates with a positive result $\{\True\}$.
\qed\end{proof}

% \begin{theorem}
% Let $S\equiv\exists x\in B\,\, (f^1=0\,\wedge g^1\geq 0)\,\vee\ldots\vee\,(f^u=0\,\wedge g^u\geq 0)$ be a formula in
% $\mathcal{B}$ containing the free variables $p\in P$. Then the algorithm $\mathrm{SoEI\_disj}(S,P,r)$ 
% fulfills its specification.
% \end{theorem}
% \begin{proof}
%   The proof of correctness of the returned result is easy. For definiteness,
%   denote by $S^i:\equiv(\exists x\in B \;.\; f^i=0\,\wedge g^i\geq 0)$.  If the
%   sentence is robustly false for all $p_0\in P$, then for all $i$, $S^i(p_0)$ is
%   robustly false and the second property of the specification follows from the
%   fact that $\mathrm{SoEI}$ satisfies this property.  Assume that for all
%   $p_0\in P$, $S(p_0)$ is robustly true. It follows from the definition of
%   robustness that at for at least one $i$, $S^i(p_0)$ is robustly false. There
%   exists a neighborhood $U(p_0)$ and an $\varepsilon(p_0)$ such that if $P'\subset
%   U(p_0)$ has width smaller then $\varepsilon(p_0)$, then
%   $\mathrm{SoEI}(f^i,g^i,B,P',\varepsilon(p_0))$ terminates with $P$. Again, we
%   choose a finite subcovering $\{U(p_j))\}$ of $\{U(p);\,p\in P\}$ and select
%   $\varepsilon$ to be the minimal $\varepsilon(p_j)$.  
% \qed\end{proof}

\subsection{Universal quantifiers}
\begin{theorem}
Let $S$ be a formula and let $I$ be a closed interval. Let $P$ be an $l$-box bounding the free variables of the formula $\forall x\!\in\! I \;.\; S$. Assume that an algorithm $\mathrm{CheckSat}$ fulfilling the definiteness
property is given.
Then also the algorithm $\mathrm{Univ}(\forall x\!\in\! I \;.\; S,P,r)$ described in Section~\ref{Univ}
fulfills the definiteness property.
\end{theorem}
\begin{proof}
Assume that for all $p_0\in P_0$, the sentence $\forall
x\in I\;.\; \Subst{S}{p}{p_0}$ is robustly true.  Then, by
Lemma~\ref{robusttrue}, for all $p_0\in P_0$ and all $x_0\in I$,
$\Subst{\Subst{S}{p}{p_0}}{x}{x_0}$ is robustly true and the property follows
directly from the assumption on $\mathrm{CheckSat}$.

Assume now that for all $p_0\in P_0$, $\forall x\in I\;.\; \Subst{S}{p}{p_0}$ is
robustly false.  Let $p_0\in P_0$. Then there exists a $x_0\in I$ such that
$\Subst{\Subst{S}{p}{p_0}}{x}{x_0}$ is false, and hence, due to
Lemma~\ref{robustfalse}, it is also robustly false. From this,
Lemma~\ref{lem:robust_neighborhood} implies that there is a neighborhood $P(p_0)$
of $p_0$ and $I_0$ of $x_0$ such that for all $p_0'\in P(p_0)$ and $x_0'\in I_0$,
$\Subst{\Subst{S}{p}{p_0'}}{x}{x_0'}$ is false.  It follows from the assumption
on $\mathrm{CheckSat}$ that there exists an $\varepsilon_{p_0}$ such that for all
$P'\subseteq P(p_0), I'\subseteq I$ of width at most $\varepsilon_{p_0}$,
$\mathrm{CheckSat}(P'\times I',S,\varepsilon_{p_0})$ terminates with
$\{\False\}$. 

Because $P_0$ is compact, we can cover it by $\{P(p_0);\,p_0\in\Lambda\}$ for a finite set $\Lambda$. 
It is easy to see that there exists an $\varepsilon'$ such that any box of side-length smaller than $\varepsilon'$ is in at least one of these $P(p_0)$. 
Now, choose $\varepsilon$ to be smaller than $\varepsilon'$ and smaller than $\varepsilon_{p_0}$ for all $p_0\in\Lambda$. For any box $P$ of side-length at most $\varepsilon$, the algorithm
$\mathrm{Univ}(\forall x\in I\;.\; S, P,\varepsilon)$
terminates with $\{\False\}$. \qed
\end{proof}

\Long{If we used a dual algorithm to $\mathrm{Univ}$ for handling additional
existential quantifiers (in addition to those from the base case of our class
$\mathcal{B}$), the definiteness part of the proof would not go through in the 
case of a robustly true input (corresponding to the robustly false case for
universal quantification). If $\exists x\in I\;.\; S$ is robustly true, then
this does \emph{not} imply that there exists a $x_0\in I$ such that
$\Subst{S}{x}{x_0}$ is robustly true. The topological degree was a way to get
around this problem in the case where the number
of variables is the same as the number of equations. }

\subsection{Conjunction and Disjunction}

\begin{theorem}
  Let $S$ and $T$ be two formulas in $\mathcal{B}$  and assume that $\mathrm{CheckSat}$ fulfills the definiteness property both when applied to $S$ and when applied $T$. Then $\mathrm{Conj}(S\wedge T, P, r)$ (described in Section~\ref{cd}) also fulfills the definiteness property.
\end{theorem}

\begin{proof}
  Let $p_S$, and $p_T$, respectively, be the function that projects any
  $l$-tuple corresponding to the free variables of $S\wedge T$ to those
  components corresponding to the free variables of $S$, and $T$, respectively.

We first assume that for all $p_0\in P_0$
the sentence $\Subst{(S\wedge T)}{p}{p_0}$ is
robustly true. Then for all $p_0\in P_0$, 
$\Subst{S}{p_S(p)}{p_S(p_0)}$ is robustly true and for all $p_0\in P_0$,
$\Subst{T}{p_T(p)}{p_T(p_0)}$ is robustly true. So, by definiteness of the recursive call, there exists an
$\varepsilon_1>0$ such that if $r\leq\varepsilon_1$ and the width of
$P_{1}\subseteq P_0$ is less than $r$, then $\mathrm{CheckSat}(S, p_S(P_{1}),r)$
terminates with $\{ \True \}$. 
An analogous $\varepsilon_2$ exists for $T$.  For $\varepsilon<\min \{\varepsilon_1, \varepsilon_2\}$,
$r\leq\varepsilon$ and $P\subseteq P_0$ of width less than $r$, 
$\mathrm{Conj}(S\wedge T,P,r)$ terminates with $\{ \True \}$.

Suppose that for all $p_0\in P_0$, $\Subst{(S\wedge T)}{p}{p_0}$ is robustly false.  Then, for any $p_0\in P_0$,
either $\Subst{S}{p_S(p)}{p_S(p_0)}$ or $\Subst{T}{p_T(p)}{p_T(p_0)}$ is robustly false. Let $p_0\in P_0$.
Assume, without loss of generality, that $\Subst{S}{p_S(p)}{p_S(p_0)}$ is robustly false.
By Lemma~\ref{lem:robust_neighborhood} there exists a neighborhood $U$ of $p_0$
such that for every $u\in U$, $\Subst{S}{p_S(p)}{p_S(u)}$ is robustly
false.  Let $P(p_0)\subseteq P$ be a box neighborhood of $p_0$ contained in the interior of $U$.
By assumption, there exists an $\varepsilon_{p_0}>0$
such that if $r\leq \varepsilon_{p_0}$ and $P'\subseteq P(p_0)$ has width at most $r$, then
$\mathrm{CheckSat}(S, p_S(P'), r)$ terminates with $\{\False\}$, hence $\mathrm{Conj}(S\wedge T, P', r)$ terminates
with $\{\False\}$ as well.

This can be done for each $p_0\in P_0$. Let $P(p')^\circ$ be the interior of $P(p')$ in the topology of the box $P_0$. 
Then $\{P(p')^\circ\,|\,p'\in P_0\}$ is an open cover of the compact space $P_0$ and there exists a finite subcovering 
$\{P(p_1)^\circ,\ldots, P(p_m)^\circ\}$ of $P_0$. Take $\varepsilon$ to be so small that each box $P\subseteq P_0$ of width
at most $\varepsilon$ is contained in some $P(p_j)$ and $\varepsilon<\min_i \epsilon_{p_i}$. Then 
$\mathrm{Conj}(S\wedge T, P, r)$ terminates with $\{\False\}$ for any $r\leq\varepsilon$ and any box $P$ of width
at most $r$.
\end{proof}

For disjunctions the situation is analogous. 

Together with the correctness proof from Section~\ref{sec:correctness} this concludes the proof of Theorem~\ref{thm:main}.

\section{Limitations on Generalization}
\label{sec:nonquasidec}

%\st{In this section we prove Theorem}~\ref{thm:nonquasidec} \st{that states that it is not possible to remove the restriction on the number of variables versus number of equations in the definition of class $\mathcal{B}$.}

We showed in Lemma~\ref{m<n} that an overdetermined system of equations ($m<n$) never has a robust solution. In the underdetermined case ($m>n$), in some cases, we could fix $m-n$ input variables in $f$ to constants $a\in\R^{m-n}$ and try to analyze the formula $\exists x\in \bar\Omega^a
\;.\; f(a,x)=0$, where $\bar\Omega^a=\{x\in\R^n \, | \,(a,x)\in \bar\Omega\}$.
If $f(a,\cdot )$ has a robust zero in $\bar\Omega^a$, then $f$ has a robust zero in $\bar\Omega$.
However, the converse is not true: If $f(a,\cdot)$ does not have a robust zero in $\bar\Omega^a$ for any
fixed choice of $a\in\R^{m-n}$ (the components of $a$ ranging over all $(m-n)$-subsets of the
total number of $m$ variables), $f$ still may have a robust zero in $\bar\Omega$. 
\Long{For an example, consider the Hopf fibration $H:S^3\to S^2$ and define a map
$\tilde{h}:\R^4\to \R^3$ in polar coordinates by $\tilde{h}(rx)=rH(x)$ for $r\geq 0$ and $x\in S^3$.
Let $h$ be the restriction of $\tilde{h}$ to the box $B=[-1,1]^4$. Clearly, $0$ is the only zero of 
$h$. For any $a\in [-1,1]$, $h^a(x)=h(a,x)$ has not a robust zero in $[-1,1]^3$ (for $a=0$, a small perturbation
of $h^a$ is nowhere zero). 
However, $\exists x\in B \;.\; h=0$ is robustly true. To show this, assume it is not and let $h_1$ be a nowhere zero
$1$-perturbation of $h$. Then $F(t)=th+(1-t)h_1$ is a homotopy between $h$ and $h_1$, $0\notin F(t)(\partial B)$ for all $t$.
So, the map $H_1:=h_1/|h_1|$ from $S^3$ to $S^2$ is homotopic to the Hopf map $H=h/|h|:S^3\to S^2$.
Further, $H_1$ is homotopic to the trivial map via $G(t): x\mapsto h_1(tx)/|h_1(tx)|$, but $H$ is not homotopically
trivial. This is  a contradiction and therefore, $h$ contains a robust zero in $B$.}

Indeed, Theorem~\ref{thm:nonquasidec} states that a generalization to the underdetermined case is (under certain weak conditions) impossible, and we will spend the rest of this section to prove this theorem.
If $Q$ is a quasi-decision procedure (Def.~\ref{def:quasidec}) and $\I$ an algorithmic assignment of $\I(f)$ to all function symbols $f$, we will denote by $Q_\I$ the algorithm that takes a sentence $\varphi$ and returns $Q(\varphi,\I)$. We need the following:
\begin{lemma}
\label{lem:strict}
Assume that there exists a quasi-decision procedure $Q$ for some class of formulas such that
each function symbol appears in each formula at most once, and such that each term in each formula consists of one single function symbol. Assume that the quasi-decision procedure has access only to the oracle $\I(f)$ for each function symbol $f$ 
in the formula.\footnote{That is, it may call $\I(f)$ with any input an arbitrary number of times, but apart from the results of calling $\I(f)$ it does not use any properties of $f$, nor does it analyze how $\I(f)$ is computed.} 
%Further, assume that the language is rich enough so that for each function symbol $f$ and $\epsilon>0$, it contains a function $\tilde{f}$ closer to $f$ than $\epsilon$ in the sup norm.

Then there exists an algorithmic assignment $\I'(f)$ to all function symbols $f$ such that $Q_{\I'}(\varphi)$ terminates \emph{if and only if} $\varphi$ is robust. 
\end{lemma}

\begin{proof}
Let us define the addition of boxes naturally by $B_1+B_2:=\{b_1+b_2:\,b_1\in B_1,b_2\in B_2\}$. 
For every function symbol $f$ corresponding to a function  $B\to\R^n$, let $\I'(f)$ be the algorithm defined by $\I'(f)(B'):=\I(f)(B')+[-\Width{B'},\Width{B'}]^n$.
This algorithm is a modification of $\I(f)$, it still represents the function $f$
and satisfies the assumptions of Definition~\ref{def:intcomp}. However, for any box $B'\subseteq B$, the output 
$\I'(f)(B')$ contains $f(B')$ in its interior. 
We will show that $Q_{\I'}$ terminates if and only if the input is robust.

By definition of $Q$, $Q_{\I'}(\varphi)$ terminates whenever $\varphi$ is a robust sentence. 
It remains to prove that it does not terminate for inputs that are not robust. 
Let $\varphi$ be a fixed non-robust sentence.  
For proving that $Q_{\I'}(\varphi)$ does not terminate, we assume that it does terminate and derive a contradiction. 

$Q_{\I'}(\varphi)$ only uses a finite number of evaluations of $\I'(f)(B)$ with $f$ being a function in $\varphi$. Let $\tilde{\varphi}$ be a perturbation of $\varphi$ in which each function $f$ is replaced by an interval computable function $\tilde{f}$ representable by a term in our first-order language such that 
\begin{itemize}
\item $\tilde{\varphi}$ and $\varphi$ have different truth values,
\item $\tilde{f}(B)\subseteq \I'(f)(B)$ for every $B$ used by $Q_{\I'}(\varphi)$ in a call to $\I'(f)(B)$. 
\end{itemize}
Such functions exist, because $\varphi$ is non-robust, $\I'(f)(B)$ contains $f(B)$ in its interior and arbitrarily close to $f$ are other functions representable by a term in our first-order language.  
Now, let $\I''$ be equal to $\I$ with the exception that for every $\tilde{f}$ occurring in $\tilde\varphi$, $\I''(\tilde{f})(B):=$
\begin{itemize}
\item $\I'(f)(B)$, for every box $B$ used by $Q_{\I'}(\varphi)$ in a call to $\I'(f)(B)$, and
\item $\I(\tilde{f})(B)$, otherwise.
\end{itemize}
$\I''$ still satisfies both axioms of Definition~\ref{def:intcomp}. 
All function symbols in both $\varphi$ and $\tilde{\varphi}$ are mutually different and both $Q_{\I'}$ and $Q_{\I''}$ do not use 
any other information about the function symbols in $\varphi$ and $\tilde{\varphi}$ than the evaluations $\I'(f)$ and $\I''(\tilde{f})$, respectively. 

However, for every call $\I''(\tilde{f})(B)$ of $Q_{\I''}(\tilde{\varphi})$, and corresponding call $\I'(f)(B)$ of $Q_{\I'}(\varphi)$, $\I''(\tilde{f})(B)=\I'(f)(B)$.
%However, $\I''(\tilde{f})(B)=\I'(f)(B)$ for every call used by $Q_{\I''}(\tilde{\varphi})$, resp. $Q_{\I'}(\varphi)$, 
Hence $Q_{\I''}(\tilde{\varphi})$ uses exactly the same information about its input as $Q_{\I'}(\varphi)$ and they have to return the same result.  But this is impossible, because $\varphi$ and $\tilde{\varphi}$ have different truth values.

Therefore, $Q_{\I'}$ does not terminate whenever the input is non-robust. \qed
\end{proof}
  
For proving Theorem~\ref{thm:nonquasidec} we use a reduction from a~recent undecidability result~\cite[p. 19]{Franek:2014}. For this we introduce the following notions: A \emph{triangulation} of the box $[0,1]^d$, with $d\in\N$, is a subdivision of $[0,1]^d$ into a finite set $S$ of simplices such that the intersection of any two simplices in $S$ is again a simplex (possibly empty) in $S$. A \emph{piecewise linear function} from $[0,1]^d$ to $\R^d$ is a function that is linear on each simplex of some triangulation. It is uniquely determined by values on the vertices of the simplices. If the simplices have rational coordinates and the values of $f$ on the vertices are all rational, then $f$ is interval computable; moreover, for any box $B\subseteq[0,1]^d$ with rational vertices,
the image $f(B)$ can be computed exactly by means of linear programming. We summarize the statement given in~\cite[Inequalities, Section 4]{Franek:2014}.
\begin{theorem}
\label{thm:non-dec}
There is no algorithm with the following specification: \pagebreak[2]
\begin{description}
\item[Input: ]~
  \begin{itemize}
  \item $n,d\in\N$,
  \item $T$, a triangulation of $[0,1]^d$ with rational vertices
  \item $(f, g):[0,1]^d\to\R^n\times\R$, piecewise linear with rational values on vertices of $T$
  \end{itemize}
\item[Output: ] At least one correct answer from the following two options:
\begin{itemize}
\item $\exists x\in[0,1]^d\;.\; f(x)=0\,\wedge\,g(x)\leq 0$ is robustly true,
\item Some $1$-perturbation of $\exists x\in [0,1]^d\;.\; f(x)=0\,\wedge\,g(x)\leq 0$ is false.
\end{itemize}
\end{description}
\end{theorem}
In the cited theorem, the notion of ``robustly true'' means that for some $\epsilon>0$, for \emph{arbitrary continuous} 
functions $\tilde{f}_i$ and $\tilde{g}$ such that $\|\tilde{f}_i-f\|<\epsilon$ and $\|\tilde{g}-g\|<\epsilon$, it holds that the sentence 
$\exists x\in[0,1]^d\;.\; \tilde{f}(x)=0\,\wedge\,\tilde{g}(x)\leq 0$ is true, not only for interval computable functions from a specified language. However,
if our first-order language contains all piecewise linear functions with rational values on rational vertices, then both notions of robustness are equivalent. This can be shown as follows:

Assume that all piecewise linear $\epsilon$-perturbations $\tilde{f}_i^{PL}, \tilde{g}^{PL}$  of $f_i, g$ satisfy $\exists x\in [0,1]^d\;\; \tilde{f}^{PL}(x)=0\,\wedge\,\tilde{g}^{PL}(x)\leq 0$, and for some continuous $(\epsilon/2)$-perturbations
$f_i', g_i'$ the sentence $\exists x\in [0,1]^d\;\; {f'}(x)=0\,\wedge\,{g'}(x)\leq 0$ is false. 
Then the last sentence is also ``robustly false'' by the remarks after Lemma~\ref{robustfalse}: ``robustly false'' here means that
any small enough continuous perturbation is false (note that $f'$ and $g'$ doesn't need to be interval computable). 
However, arbitrary close to $f'$ and $g'$ are some piecewise linear functions, which contradicts our assumption
that any piecewise linear $\epsilon$-perturbation of $\exists x\in [0,1]^d\;\; {f}(x)=0\,\wedge\,{g}(x)\leq 0$ is true. 
Therefore, both notions of being robustly true are equivalent and we do not need to distinguish them further.

Further, Theorem~\ref{thm:non-dec} still holds, if we assume that the function symbols $\{f_1,\ldots, f_n, g\}$ in the input formula
\begin{equation}
\label{e:soei}
\exists x\in [0,1]^d\,\,f_1=0\,\wedge\,\ldots\,\wedge f_n=0\,\wedge\,g\leq 0
\end{equation} are all pairwise different, and that the perturbations
consist of formulas in which all functions are pairwise different. This can be seen as follows:

If, in formula (\ref{e:soei}), two functions $f_i$ and $f_j$ or $f_i$ and $g$ coincide, we can easily construct
an arbitrary small perturbation of (\ref{e:soei}) that is false, because each component can be perturbed independently.
So, without loss of generality, we may assume that all function symbols in the input of Theorem~\ref{thm:non-dec} are different.
It can easily be shown that the sentence (\ref{e:soei}) is robustly true if and only if for some $\epsilon>0$, each
$\epsilon$-perturbation 
$$
\exists x\in [0,1]^d\,\,\tilde{f}_1=0\,\wedge\,\ldots\,\wedge \tilde{f}_n=0\,\wedge\,\tilde{g}\leq 0
$$
in which all $n+1$ functions $\tilde{f}_j, \tilde{g}$ are different, is true. Similarly, some $1$-perturbation is false, if some $1$-perturbation
in which all functions are different, is false. Summarizing the previous paragraphs, we obtain the following consequence:
\begin{lemma}
\label{l:nonrepeating}
Assume that we have a language containing function symbols for all piecewise linear functions on rational triangulations of $[0,1]^d$
with rational values on vertices, and the class of all sentences $\mathcal{A}$ of the type (\ref{e:soei}) 
such that in each sentence, all function symbols are different.
Then there is no algorithm with the following specification:
\begin{description}
\item[Input: ] ~
  \begin{itemize}
  \item A sentence $\varphi$ from $\mathcal{A}$.
  \end{itemize}
\item[Output: ] At least one correct answer from the following two options:
\begin{itemize}
\item $\varphi$ is robustly true wrt. the class $\mathcal{A}$
\item Some $1$-perturbation of $\varphi$ from $\mathcal{A}$ is false.
\end{itemize}
\end{description}
\end{lemma}
Now we are ready to prove Theorem~\ref{thm:nonquasidec}:

\begin{proof}[of Theorem~\ref{thm:nonquasidec}.]
We will assume that a quasi-decision procedure for the class of sentences defined in Theorem~\ref{thm:nonquasidec} exists, and derive a contradiction. Let us call the assumed quasi-decision procedure $Q$. We prove that the existence of $Q$ implies the existence of an algorithm solving the undecidable problem from Lemma~\ref{l:nonrepeating}. For this we will first (Step~1) construct an algorithm computing positive information, then (Step~2) an algorithm computing negative information, and finally (Step~3) run them in parallel to get an algorithm specified in Lemma~\ref{l:nonrepeating}.

{\bf Step 1.} First we show that the existence of $Q$ implies the existence of an algorithm with input as in Lemma~\ref{l:nonrepeating} such that it terminates iff $\exists x\in[0,1]^d\;.\; f=0\,\wedge\,g\leq 0$ is robustly true.

We can easily construct an algorithm assigning to each piecewise linear function $f:[0,1]^d\to \R$ with rational values on the vertices 
an algorithm $\I(f)$ satisfying
the axioms in Definition~\ref{def:intcomp}. 
From the quasi-decision procedure $Q$ for general systems of equations and inequalities we get an algorithm $Q_\I$ 
that takes an input such as in Lemma~\ref{l:nonrepeating} and decides whether it is robustly true or not, whenever
the input is robust.
By Lemma~\ref{lem:strict}, we can algorithmically replace $\I$ by $\I'$ and obtain an algorithm $Q_{\I'}$ 
that terminates iff the input is robust. This procedure can be modified such that 
% Assume that $Quasi$ is a quasi-decision procedure with input $n,d,\I(f)$ and $\I(g)$ that terminates if $\exists x\in [0,1]^d\,f=0\,\wedge\,g\leq 0$ is robust. Let $f, g$ be piecewise linear functions on a rational triangulation $T$ of $[0,1]^d$ with rational values on vertices. Then we can
% algorithmically construct interval-computable representations $\I(f), \I(g)$ and run $Quasi'(\I(f), \I(g))$. 
% Then the algorithm $Quasi'$ defined in Step 1 terminates iff the input is robust. 
instead of terminating with $\{\False\}$, it runs forever. The result is an algorithm that terminates if and only if the input $\exists x\in [0,1]^d\,f=0\,\wedge\,g\leq 0$ is robustly true.

{\bf Step 2.} Now we show that there exists an algorithm with input such as in Lemma~\ref{l:nonrepeating} with the following specification:
\begin{itemize}
\item if some $1/2$-perturbation of $\exists x\in [0,1]^d\;.\; f=0\,\wedge\,g\leq 0$ is false, then it terminates, and
\item if it terminates, then some $1$-perturbation of the above formula is false.
\end{itemize}
%\marginpar{\Peter{Where is Theorem~\ref{thm:nonquasidec} used here?} \Stefan{Do you mean Theorem~\ref{thm:non-dec}? That is not used in Step 2, just in Step 1 and hence 3. Theorem~\ref{thm:nonquasidec} is what we are proving.}} 
This algorithm can be described as follows: In the $i$-th step, it constructs the $i$-th barycentric subdivision $T^{(i)}$ of the given triangulation $T$, and further constructs all piecewise linear functions $f', g'$ on this subdivision such that their values on the vertices $\{v_k\}_k$ of $T^{(i)}$ are rational with denominators at most $i$ and such that for each $k$, the values $f_i'(v_k)$ resp. $g'(v_k)$ differ from $f(v_k)$ resp. $g(v_k)$ by less than $1$. For all such piecewise linear functions $f',g'$, the truth value of $\exists x\in[0,1]^d\;.\; f'=0\,\wedge\,g'\leq 0$ can be computed. Moreover, due to the restriction on the denominators of the values on the vertices, there exists only a finite number of such functions. So, for all those finitely many $f'$ and $g'$, the algorithm checks whether $\exists x\in [0,1]^d\;.\; f'=0\,\wedge\,g'\leq 0$ is false and terminates as soon as it finds a pair $(f', g')$ for which the formula is false.

In the rest of step 2 of the proof we show that this algorithm satisfies the above specification.

The absolute value $|f_i-{f'}_i|$ of a linear function $f_i-{f'}_i$ is a convex function on each simplex $\Delta\in T^{(i)}$, so on each simplex it attains its maximum on a vertex. Therefore, a piecewise linear function ${f'}$ is a $1$-perturbation of $f$ iff its restriction to the vertices is a $1$-perturbation of the restriction of $f$. It follows that $f'=0\,\wedge\,g'\leq 0$ is a $1$-perturbation of $f=0\,\wedge\,g\leq 0$ if and only if the differences  $|f_i(v_k)-\tilde{f}_i(v_k)|\leq 1$ and $|g_i(v_k)-\tilde{g}_i(v_k)|\leq 1$ for all $i$ and all vertices $v_k$. Assume that $f,g: [0,1]^d\to\R^n\times\R$ are piecewise linear on a given triangulation $T$ of $[0,1]^d$ and that  some $1/2$-perturbation of $f=0\,\wedge\,g\leq 0$ is unsatisfiable. Each continuous function can be approximated arbitrarily precisely by some piecewise linear function on an iterated barycentric subdivision.
So, there exists an iterated barycentric subdivision $T^{(i)}$ of $T$ and piecewise linear functions 
$\tilde{f}, \tilde{g}$ on $T^{(i)}$  such that $\exists x\in[0, 1]^d\;.\;\tilde{f}=0\,\wedge\,\tilde{g}\leq 0$ is a false $1$-perturbation of $\exists x\in [0,1]^d\;.\; f=0\,\wedge\,g\leq 0$. The algorithm finds this perturbation in its $i$th step and terminates.

Conversely, if the algorithm terminates, then it had found a false $1$-perturbation of 
$\exists x\in [0,1]^d\;.\; f=0\,\wedge\,g\leq 0$.

{\bf Step 3.} Finally, we show that the existence of $Q$ contradicts Lemma~\ref{l:nonrepeating}. Given piecewise linear functions $(f,g): [0,1]^d\to\R^n\times \R$ with non-repeating function symbols and the quasi-decision procedure $Q$ for systems of equations and inequalities, we could run an algorithm specified in Step~1 that terminates if and only if $\exists x\in [0,1]^d\;.\; f=0\,\wedge\,g\leq 0$ is robustly true. Further, by Step 2, we could run another algorithm that terminates whenever some $1/2$-perturbation of this formula is false. A formula is either robustly true, or has a false 1/2-perturbation, so at least one of these algorithm would always terminate. If the first algorithm terminates,
we know that the formula is robustly true and if the second one terminates, we know that some 1-perturbation is false. 
Thus, we could choose at least one correct answer from the output specified in~Lemma~\ref{l:nonrepeating},
which is impossible. \qed
\end{proof}

\section{Related Work}
\label{sec:related-work}

From the very beginning of engineering the notion of robustness has played a
key role. This is being recognized more and more in several scientific fields:
For example, the field of robust control~\cite{Zhou:97,Bhattacharyya:95} is now
considered as a central subject of control engineering. Robustness also plays an
increasingly important role in applied and computational mathematics, as shown by
the emerging fields of robust optimization~\cite{Ben-Tal:09} and uncertainty
quantification (with a journal of the same name recently having been launched by SIAM).

Also in the computing field, robustness has been a core issue from the very
beginning. In computer systems design this is usually captured by the keyword
''fault-tolerance'' and for numerical algorithms ''stability''.  Robustness also
plays an important role in computational geometry~\cite{Yap:04}. 

In the complexity analysis of algorithms, the notion of perturbation has helped to
explain the good practical behavior of algorithms with exponential worst-case
complexity~\cite{Spielman:04}. The present paper in analogy applies Spielman and
Deng's~\citeyear{Spielman:09} motivation ''The basic idea is to identify typical
properties of practical data, define an input model that captures these
properties, and then rigorously analyze the performance of algorithms assuming
their inputs have these properties'' to undecidable problems, where the main
goal then is not performance analysis but finding a terminating algorithm.

Apparently, the first paper that follows this approach of ensuring termination of an algorithm for all robust inputs to an undecidable problem (in this case safety
verification of hybrid systems) is due to
Fr{\"a}nzle~\citeyear{Fraenzle:99}. Since then, a similar approach has been applied
several times~\cite[e.g.]{Fraenzle:99,Fraenzle:01,Ratschan:02b,Ratschan:02f,Damm:07} to
problems in formal verification.

To the best of our knowledge, the first paper to apply such an approach to decision procedures for the real numbers is by one of the co-authors~\cite[Theorem~5]{Ratschan:01c} (see also~\cite[Theorem~6]{Ratschan:02f}), based on an analysis of robustness of
first-order formulas~\cite{Ratschan:02b}. The main difference to the present paper and---at the same time---main weakness is, that it expresses equalities of the form $f(x)=0$ as a conjunction of two equalities of the form $f(x)\leq 0\wedge -f(x)\leq 0$ which---in general---loses robustness, since the two occurrences of $f$ can be perturbed independently and a solution of $f(x)=0$ can vanish under perturbations of $f(x)\leq 0\wedge -f(x)\leq 0$. Hence, the corresponding algorithm need not necessarily terminate in such cases of satisfiable equalities.

Recently, Gao et. al.~\citeyear{Gao:12} took a similar approach: They model perturbations of formulas by the notion of $\delta$-strengthening which roughly means that  inequalities of the form $f\geq 0$ are replaced by inequalities of the form $f \geq \delta$, where $\delta>0$. However, instead of allowing non-termination in non-robust cases, the approach uses the notion of $\delta$-decidability that requires an algorithm to terminate always, but either decides that the input formula is true, or that a $\delta$-strengthening of the input formula is false. These two answers overlap, and especially for non-robust inputs, for every $\delta>0$, a $\delta$-strengthening of the input is false, and hence both answers are allowed. 

Since $\delta$-decidability cannot give a definite answer for inputs that are false, $\delta$-decidability does not imply quasi-decidability, in general. 
However, it does imply quasi-decidability for classes of formulas that are closed under negation, since then it is possible to run the corresponding algorithm in parallel on both the input formula and its negation. It would be an easy extension of the algorithm in this paper to return also quantitative information on robustness (i.e. a value $\varepsilon\in\R$ s.t. the input is $\varepsilon$-robust). 

Gao and co-authors handle equalities of the form $t=0$ as the non-robust formula $-|t|\geq 0$. Hence---in contrast to the present paper---their approach cannot prove equalities to have a solution. For example, it cannot prove that $\exists x\in[-10, 10]\;.\; x=0$ is true since this formula is  handled as the non-robust formula $\exists x\in[-10,10]\;.\; -|x| \geq 0$ for which, for every $\delta>0$, the $\delta$-strengthening  $\exists x\in[-10,10]\;.\; -|x| \geq \delta$ is false. And indeeed, in such cases the approach  returns that a $\delta$-strengthening of the formula is false. The paper~\cite{Gao:12} also studies complexity of such algorithms in some model of computable analysis~\cite{Brattka:08}. Another  paper~\citeyear{Gao:12a} studies $\delta$-decidability in a satisfiability modulo theory (SMT) context, where the approach either returns ``unsatisfiable'' or ``a $\delta$-weakening is satisfiable'' with the notion of $\delta$-weakening defined in analogy to $\delta$-strengthening. 

Due to the fact that those approaches~\cite{Ratschan:02f,Gao:12,Gao:12a} do not handle equalities directly, but reformulated as inequalities, those algorithms that do not need to, and in fact do not exploit continuity of the involved functions. In contrast to that, in the present paper we use the topological degree as the notion that captures the essential information about the roots of continuous functions under continuous perturbations. 

All those approaches can depend on the precise way perturbations of first-order formulas are modeled. There are various possibilities for this, and some have been compared~\cite{Ratschan:02b}, but a comprehensive exploration of this is still missing.

The approach of relaxing the semantics of first-order formulas can be taken even further than just relaxing the dichotomy satisfiable/unsatisfiable. For example, one can weaken the necessity of distinguishing between close values~\cite{Casagrande:09}, or introduce quantifiers that are weaker than the classical ones~\cite{Ratschan:03}.

Collins~\citeyear{Collins:08b} presents similar result to ours for the special case of systems of $n$ equalities in $n$ variables, formulated in the language of computable analysis~\cite{Brattka:08}. However, the paper contains only very rough proof sketches, that we were not able to complete into full proofs.

%Recent advances in computational topology have produced nice implications for algorithmic solvability of real systems of $n$ equations in $m$ variables that are exposed to noise or perturbations. 
Franek and Kr{\v c}{\'a}l study~\cite{Franek:2014} the problem whether or not each continuous $r$-perturbation of a system $f(x)=0$ has a solution or not, where $f: K\to\R^n$ is a piecewise linear function defined on a finite simplicial complex $K$. This turns out to be decidable whenever $\dim K< 2n-2$ or $n$ is even and undecidable for a fixed odd $n\geq 3$ and arbitrary $K$.

Verification of zeros of systems of equations is a major topic in the interval
computation community~\cite{Neumaier:90,Rump:98,Kearfott:02,Frommer:05}. However, here
people are usually not interested in some form of completeness of their methods,
but in usability within numerical solvers for systems of equations or global
optimization.

Basic existence theorems that are commonly used for proving that an equation
$f=0$ has a solution in $B$ are Kantorovich's, Miranda's and Borsuk's
theorem. Among these Borsuk's theorem is
the strongest~\cite{Alefeld:04,Frommer:05}, that is, if the assumptions of the other theorems are fulfilled,
then the assumptions of Borsuk's theorem are fulfilled as well.

We will now recall Borsuk's theorem and then compare its power for proving existence of a zero with that of the use of the topological degree:

\begin{theorem}[Borsuk's theorem]
If $B\subseteq\R^n$ is open, bounded, convex and symmetric with respect to an interior point $x$, $f: \bar{B}\to\R^n$ is continuous 
and non-zero on the boundary $\partial B$ and if for any $x+y\in\partial B$ and $\lambda>0$, $$f(x+y)\neq \lambda f(x-y),$$ 
then $f=0$ has a solution in $B$.
\end{theorem}

It can be shown that if the assumption of Miranda's theorem are satisfied, then the degree has to be $1$ or $-1$
and if the assumption of Borsuk's theorem are satisfied, then the degree $\deg(f,B,0)$ has to be an odd number\footnote{This
can be shown as follows. The function $\tilde{f}:= f/|f|: \partial B\to S^{n-1}$ is homotopic to 
$g(x):=\frac{\tilde{f}(x)-\tilde{f}(-x)}{|\tilde{f}(x)-\tilde{f}(-x)|}$ via the homotopy 
$H(t,x)=\frac{\tilde{f}(x)-t\tilde{f}(-x)}{|\tilde{f}(x)-t\tilde{f}(-x)|}$, so $\tilde{f}$ and $g$ have the same degree.
Assumptions on $B$ imply that $\partial B\simeq S^{n-1}$ and an odd map $g(-x)=-g(x)$ between spheres 
has odd degree~\cite[p. 180]{Dieudonne:2009}.}.
On the other hand, if $f$ has an isolated zero of even degree, then one cannot prove that using
Borsuk's theorem.
A simple illustration of this is the complex function $f(z)=z^2$ from $\C\simeq\R^2$ to itself,
\Long{
$$
f:
\begin{pmatrix}
x\\ y
\end{pmatrix}\mapsto
\begin{pmatrix}
x^2-y^2\\ 2xy
\end{pmatrix}
$$}
defined in a~symmetric and convex neighborhood $B$ of $0\simeq (0,0)$. 
This function has a robust zero in $B$ and $\deg(f,B,0)=2$, 
so the assumptions of Borsuk's theorem are not fulfilled in any such neighborhood $B$.

% Still, our completeness result implies that the algorithm that we use is
% stronger then the known common algorithms for detecting zeros, like
% Miranda's and Borsuk's theorem.

An essential ingredience of our algorithm is the computation of the topological degree. 
Many papers deal with the question of an effective implementation, e.g.~\cite{Erdelsky:73,Kearfott:79,Boult:86,Aberth:94,Franek:12b}.
Our online package {\tt TopDeg}\footnote{\url{http://topdeg.sourceforge.net}} computes $\deg(f,B,0)$ for a function $f$ defined
as an expression containing symbols such as polynomials and sin, and a~low-dimensional box $B$. 
The degree can also be computed by the use of packages for 
\emph{simplicial homology} computations,
such as Chomp\footnote{\url{http://chomp.rutgers.edu}}, GAP homology packes\footnote{\url{http://www.linalg.org/gap.html}}, or a collection of MATLAB routines PLEX~\footnote{\url{http://comptop.stanford.edu/u/programs/plex/}}. 
However, to compute the degree with the use of these programs, one has to create first a simplicial approximation of 
$f/|f|: \partial B\to S^{n-1}$, which can be done by means of interval arithmetic.

A limitation of our approach is the fact that while in the context of real-world problems and engineering applications, the robustness assumption is natural, theorems with a purely mathematical motivation often fail to be robust. In such a context, the only option to automatize theorem proving 
of first-order sentences of the reals with function symbols such as $\sin$ is the systematic usage of heuristics. This has been successfully implemented in the MetiTarski~\cite{Akbarpour:10} package.

\section{Conclusion}
\label{sec:conclusion}

Motivated by the fact that in many application domains robustness is an essential property of formal models, we showed that for an undecidable class of first-order formulas over the real numbers one can algorithmically check
satisfiability in all robust cases (under the additional assumption that all variables range over bound intervals). Moreover, we showed that it is not possible to generalize this result to the case without restrictions on the number of variables versus number of equations. Still, it might be possible to find a quasi-decision procedure for certain, specific numbers of variables versus equations. Moreover, it might be possible to find a quasi-decision procedure for a class of formulas with functions that are more specific than general interval computable.

The generalization to arbitrary existential quantification is hindered by the fact that the property that $\forall x\in I\;.\; S$ is robustly true if and only if for each $x_0\in I$, the sentence $\Subst{S}{x}{x_0}$ is robustly true (Lemma~\ref{robusttrue}) does not hold in analogy for existential quantifiers. The  sentence $\exists x\in [-1,1] \;.\; x=0$ is robustly true but for any $x_0\in [-1,1]$, the sentence $x_0=0$ ($x_0$ is considered to be a constant function here) is not robustly true. A topological reformulation of adding an existence quantifier to the beginning of a formula would be desirable and could be a subject of future research.

It also remains an open problem to come up with an algorithm that is both a quasi-decision procedure and efficient in practice.

\web{}{
\begin{acknowledgements}
  \acktext
\end{acknowledgements}}

\bibliographystyle{abbrv}
\bibliography{real_quasidec}

\end{document}